\documentclass{draft_template}

\usepackage{float}
\usepackage{multicol}
\usepackage{subcaption}
\usepackage[ruled,vlined]{algorithm2e}
\usepackage{caption}
\usepackage{tikz}
\usepackage{mathtools}
\usetikzlibrary{decorations.pathreplacing}

\usepackage{amsthm}

\usepackage{natbib}

\numberwithin{table}{section}
\numberwithin{figure}{section}
\numberwithin{equation}{section}

\usepackage[running]{lineno}

\newcommand*\patchAmsMathEnvironmentForLineno[1]{%
\expandafter\let\csname old#1\expandafter\endcsname\csname #1\endcsname
\expandafter\let\csname oldend#1\expandafter\endcsname\csname end#1\endcsname
\renewenvironment{#1}%
{\linenomath\csname old#1\endcsname}%
{\csname oldend#1\endcsname\endlinenomath}}%
\newcommand*\patchBothAmsMathEnvironmentsForLineno[1]{%
\patchAmsMathEnvironmentForLineno{#1}%
\patchAmsMathEnvironmentForLineno{#1*}}%
\AtBeginDocument{%
\patchBothAmsMathEnvironmentsForLineno{equation}%
\patchBothAmsMathEnvironmentsForLineno{align}%
\patchBothAmsMathEnvironmentsForLineno{flalign}%
\patchBothAmsMathEnvironmentsForLineno{alignat}%
\patchBothAmsMathEnvironmentsForLineno{gather}%
\patchBothAmsMathEnvironmentsForLineno{multline}%
}

\newtheorem{remark}{Remark}

\newtheorem{property}{Properties}
\newtheorem{example}{Example}
\newtheorem{theorem}{Theorem}

\newtheorem{lemma}{Lemma}
\newtheorem{definition}{Definition}

\usepackage[outdir=./]{epstopdf}
\def\layersep{3.5cm}
\SetKwFor{For}{for (}{) $\lbrace$}{$\rbrace$}

\newcommand{\ppath}{\mathcal{P}}

\begin{document}


\title{Optimal Asset Allocation  \\ For Outperforming A Stochastic Benchmark Target}

\author{Chendi Ni}
\affil{Cheriton School of Computer Science, University of Waterloo, Waterloo, N2L 3G1, Canada,
\texttt{chendi.ni@uwaterloo.ca}
}
\author{Yuying Li}
\affil{Cheriton School of Computer Science, University of Waterloo, Waterloo, N2L 3G1, Canada,
\texttt{yuying@uwaterloo.ca}
}
\author{Peter Forsyth}
\affil{Cheriton School of Computer Science, University of Waterloo, Waterloo, N2L 3G1, Canada,
\texttt{paforsyt@uwaterloo.ca}
}
\author{Ray Carroll}
\affil{Neuberger Berman Breton Hill, Toronto, M4W 1A8, Canada,
\texttt{Ray.Carroll@nb.com}
}
\date{\today}

\maketitle

\begin{abstract}
We propose a data-driven Neural Network (NN) optimization framework to determine the optimal multi-period dynamic asset allocation strategy for outperforming a general stochastic target.
We formulate the problem as an optimal stochastic control with an asymmetric, distribution shaping, objective function.
The proposed framework is illustrated with the
asset allocation problem in the accumulation phase of a defined contribution pension plan,
with the goal of achieving a higher terminal wealth than a stochastic benchmark.
We demonstrate that the data-driven approach is capable of learning an adaptive asset allocation strategy directly from historical market returns, without assuming any parametric model of the financial market dynamics. Following the optimal adaptive strategy, investors can make allocation decisions simply depending on the current state of the portfolio. The optimal adaptive strategy outperforms the benchmark constant proportion strategy, achieving a higher terminal wealth with a 90\% probability, a 46\% higher median terminal wealth, and a significantly more right-skewed terminal wealth distribution. We further demonstrate the robustness of the optimal adaptive strategy by testing
the performance of the strategy on bootstrap resampled market data, which
has different distributions compared to the training data.
\end{abstract}
\section{Introduction}
We propose a data-driven framework to compute optimal multi-period dynamic strategies
for outperforming a general stochastic benchmark target,
which is an important portfolio management problem with immediate practical applications.
There is a large extant literature on techniques for constructing
portfolios which outperform a stochastic benchmark, e.g., \citep{browne_1999_a,browne_2000,tepla_2001,basak_2006,davis_2008,
Lim_2010,Oderda_2015,Alekseev_2016,samo_2016,alaradi_2018}.

Typically, outperforming a multi-period investment benchmark is formulated as an optimal
stochastic control problem under an assumed model for trading asset price dynamics, e.g., \citep{Oderda_2015,alaradi_2018}.
In \citet{Oderda_2015}, under the assumption
that stocks follow a geometric Brownian motion and no investing constraints (i.e. infinite
leverage, trading continues if insolvent, and shorting is allowed), the authors find that a portfolio which outperforms (under certain
criteria) the benchmark market capitalization index can be constructed by a combination
of (i) the benchmark portfolio and (ii) rule-based portfolios such as equal weight and
minimum variance portfolios. The determination of the optimal weights for these portfolios
is independent of estimates of the expected returns of individual stocks. Hence this outperformance
portfolio is robust to uncertainty in the expected return parameters. A natural conjecture is that
determining asset allocation strategies that outperform a benchmark may be robust in general.

In \citet{alaradi_2018}, optimal
stochastic control techniques are also used in this context. Based on several assumptions,
\citet{alaradi_2018} formulate the control problem as a Hamilton-Jacobi-Bellman (HJB)
Partial Differential Equation (PDE), and are able to obtain a closed-form solution.
However, of necessity, this approach requires (i) the assumption of a parametric
model for the Stochastic Differential Equations (SDEs)
governing the asset price processes and (ii) no constraints on the
portfolio (i.e. infinite leverage is allowed). It is possible, in some cases, to solve the HJB PDE
numerically, and thus include more realistic constraints.

There are two main challenges in the aforementioned methods for the stochastic optimal control outperforming benchmark problem.
Firstly, unless the benchmark is specifically restricted, it can add additional stochastic state variables
in the optimal control problem \citep{alaradi_2018}. This makes solving the PDE formulated control problem numerically challenging,
due to the curse of dimensionality. Hence, this technique
is limited to a small number of stochastic factors (i.e. less than four).
Secondly, a parametric model of the
asset returns needs to be postulated, which adds challenges as the parameters can be difficult to estimate accurately.

To overcome the aforementioned challenges, in this work we use market asset return data directly to
solve a scenario-based stochastic optimal control formulation, corresponding to the original stochastic control problem.
This avoids the need to make model assumptions and parameter estimations.
In addition, we solve the stochastic optimal control problem directly,
without invoking dynamic programming to transform it into a PDE problem
(thus avoiding the curse of dimensionality).
The optimal control is represented as a neural network (NN) which is learned through training.
The features for the NN can include any state variable that influences the optimal strategy,
including the state variables associated with a stochastic target. We design a specific objective function to create a desirable terminal wealth distribution. This is done by measuring the relative performance of the strategy against an elevated final wealth of the stochastic target strategy to penalize extreme losses and limit unlikely extreme gains.

We formulate a general optimal control problem for the multi-period asset allocation portfolio which outperforms
a benchmark as an optimal stochastic control problem.
We propose a benchmark target-based
objective function which measures the difference between the terminal wealth of the
adaptive strategy and a path-dependent elevated target (which is the terminal wealth of the constant proportion strategy multiplied by a pre-defined growth factor). The objective function is
designed as a double-sided penalty function to force the terminal wealth of the adaptive strategy to be close to the elevated target.
The NN model takes three features as inputs:
the current wealth of the adaptive portfolio, the current wealth of the constant proportion portfolio, and the time remaining. In the case that the underlying assets follow simple stochastic processes, it can be shown that
the control is only a function of these variables.

Instead of formulating the problem as an HJB equation derived from dynamic programming, we
solve the single original optimal control problem directly as in \citet{li2019data}.
We define an objective function
in terms of the terminal wealth, and then solve for the control directly, using
a data-driven approach. The proposed data-driven approach does not require an estimation of the parameters of
an assumed parametric model for traded assets. We represent the control using a shallow neural
network (NN).
We remark that shallow learning is found to outperform deep learning
for asset pricing in \citet{Guetal:2018}.
We also note that good results are obtained in \citet{Hejazi2016} with an NN containing
only one hidden layer (shallow learning), in which the shallow neural network learns a good choice of distance function for efficiently and accurately interpolating the Greeks for the input portfolio of Variable Annuity contracts.

To illustrate the proposed framework, we consider a practically relevant and important problem: optimal multi-period asset allocation during the
accumulation phase of a DC pension plan.
A defined contribution (DC) plan is a retirement plan in which the employer,
employee, or both make contributions regularly with no guarantee on the accumulated
amount in the plan at the retirement date. In contrast, another type of retirement plan is the defined benefit (DB) plan, which promises to pay a set income when the employee retires.
There has been a paradigm shift from DB plans to DC plans in the United States,
Canada, the United Kingdom, and Australia, as both the public and private sectors are no longer willing to take on the risks of DB plans.

In a DC plan, the employee (investor) is often presented with a list of eligible stock and
bond funds, and then needs to specify how the DC account is to be allocated to each fund.
Typically the employee
has the opportunity to change the asset allocation at least yearly.
Normally, the DC plan is tax-advantaged,
so that there are no tax consequences triggered on rebalancing. A typical DC plan accumulation
phase would occur over 30 years, assuming a 30-year-lifetime employment
period.
The choice of the asset allocation strategy is crucial to the terminal wealth in the DC fund.

A popular asset allocation strategy for retirement plans is the constant proportion strategy,
in which the employee invests fixed proportions of the wealth into several assets.
This idea can be traced back to \citet{graham:2003}.
Among the constant proportion strategies, a very popular one is the 50/50 strategy,
in which 50\% of the wealth is allocated to stocks and 50\% of the wealth is allocated to bonds.
It is conventional wisdom that a 50/50 portfolio is an appropriate tradeoff between risk and reward
for those saving for retirement.
Although there has been a popular shift to a 60/40 portfolio (60\% in stocks) in recent years,
for illustration, we will focus on the 50/50 portfolio in this article. This would be a typical average allocation
to equities over the full accumulation phase of a lifecycle fund.\footnote{A lifecycle fund
is based on the intuitive concept of allocating a high equity weight during the early employment
years, and then moving to bonds as retirement nears. However, as shown in \citet{graf:2017},
this strategy does not outperform a constant weight strategy.}
Note that, in \citet{ForsythVetzal2018a}, it
is shown that the final wealth distributions of a constant weight allocation, and any glide path strategy having the
same average allocation as the constant weight strategy, are essentially the same. Hence there
is little to be gained by using a (deterministic) glide path compared to a constant weight strategy. Using the proposed framework to determine the optimal multi-period dynamic asset allocation strategy for outperforming a stochastic target, we address a natural and interesting question of whether it is possible to develop a dynamic allocation strategy
that outperforms the constant proportion strategy.

It is common practice in the financial industry to train and test strategy performance by splitting the historical market data path into two segments - one for training and the other for testing. We take a different approach. We aim to determine an investment strategy that would perform well statistically on a large set of data paths created through bootstrap resampling, rather than on a single historical data path. To achieve this, we generate additional data paths from the historical market data path by block bootstrap resampling of the historical data (see, e.g., \citet{politis1994stationary, politis2004automatic, patton2009correction}). Once we have a large set of price paths from bootstrap resampling, we split them into training data set and testing data set.

To demonstrate the robustness of our approach, we test the optimal adaptive strategy on market data with different distributions
from the training data. We first test the optimal adaptive strategy, learned from bootstrap resampled data
with a given expected blocksize, on bootstrap resampled data with different expected blocksizes
(thus different distributions, as noted by \citet{politis1994stationary}).
We then test the adaptive strategy learned from synthetic data generated from a parametric
jump-diffusion stochastic process (estimated from the same single historic path) on bootstrap resampled data.
Finally, we test the strategy learned on bootstrap resampling data from a segment of the historical market data
path on bootstrap resampling data generated from another non-overlapping segment of the historical data path.

To the best of our knowledge, the closest work related to the research in this paper is \citet{samo_2016},
in which the authors also use a data-driven machine learning approach for constructing
a dynamic strategy which outperforms a benchmark.
\citet{samo_2016} approximate the control by a Gaussian process and solve the optimal
hyperparameters using Bayesian inference. However, they do not assess the distributional properties of
the investment strategy, but rather evaluate the performance on a single historical path.
In addition, they only validate the performance of the strategy for a relatively short
period from 1992-2014. In contrast to our focus in this work, they consider the case of daily rebalancing
with a large number of stocks which would not be typical of a defined contribution pension plan.

In this paper, we consider investment portfolios which are combinations of a stock index and a bond index.
This type of portfolio would be typical of
a defined contribution (DC) pension plan. Our objective is to achieve a good balance of risk and return.
Specifically, we aim to shape the cumulative
distribution function of the terminal wealth so that it has desirable
risk-return characteristics. We use yearly rebalancing so that trading
costs are minimal. Our data is based on historical monthly index returns dating back to 1926, which provides a comprehensive coverage of historical financial market cycles.
To augment the data for training and testing purposes, we use
the stationary block bootstrap technique \citep{politis1994stationary}.

Furthermore, our approach differs from
\citep{samo_2016} in the learning methodology, both with respect to learning algorithms and data utilization.
Our approach can be applied to a
general multi-period asset allocation problem with few assumptions.
In addition, it can readily be scaled up to high dimensional problems (i.e. more assets and
features).
A shallow network is sufficient here, leading to a relatively small number of parameters and
computationally efficient training. In contrast to \citet{samo_2016}, we use a small number
of feature variables that only depend on the state of the adaptive portfolio and
the benchmark portfolio, rather than market-related signals. As a result,
the trading strategy is easy to interpret, practical to implement and the model is less prone to overfitting.
Furthermore, our computational results demonstrate that the optimal adaptive strategy has a higher expected terminal wealth as well as a more favorable terminal wealth distribution than the constant proportion benchmark strategy.

In summary, we make the following contributions in this research:
\begin{itemize}

\item We propose a data-driven solution to a general optimal dynamic asset allocation for outperforming a stochastic benchmark, which is formulated as a stochastic control problem. The data-driven learning bypasses the need for a robust estimation of parameters of an assumed parametric model.
In addition, closed-form solutions are only available assuming simple parametric price processes and no
portfolio constraints (i.e. infinite leverage is allowed) (see, for
example, \citet{Oderda_2015,alaradi_2018}). Existing solution techniques, which require dynamic programming, are computationally infeasible due to the high dimensionality.
In this work, we formulate the controls as the outputs of a neural network
function and avoid the curse of dimensionality of a PDE approach
\footnote{Since we consider benchmark and optimal portfolios as having two assets each,
this would result in a four-dimensional PDE problem, assuming discrete rebalancing.}.
We use a gradient-based optimization method to solve for the controls.
This approach naturally extends the method in \citep{li2019data} to the problem of outperforming a stochastic benchmark. Our philosophy
here is similar to that in \citet{samo_2016}, although our method of implementation and scope are
significantly different.

\item Unlike the commonly used one-sided quadratic shortfall objective function, we propose an asymmetric distribution shaping objective function
for the optimal asset allocation problem. The proposed objective function aims to produce an optimal dynamic and adaptive strategy which can yield significantly higher median terminal wealth than the stochastic benchmark, with only a small probability (and magnitude) of underperformance.

\item Recognizing financial data scarcity, we use block bootstrap resampling to generate both training data and testing data. We observe that the block bootstrap resampling data sets generated using different expected blocksizes lead to performance testing against different distributions.
We mathematically establish upper bounds on the probability of a training path being equivalent to a testing path to justify the soundness of the proposed stationary block bootstrap method even when the same expected blocksizes are used to generate training and testing data sets.

\item We apply the proposed data-driven framework to the allocation of the DC pension plan. In this context, the constant proportion strategy is a popular asset allocation strategy because of its simplicity in execution and its capability of diversifying market risks
effectively. However, constant proportion strategies are not able to adapt to different market
scenarios because of the predefined fixed allocations. It is a popular active research problem
within the financial industry to devise schemes that consistently outperform the constant
proportion strategy.

\item Our work has significant empirical importance and implications. The optimal adaptive asset
allocation strategy learned from the data-driven framework has a more favorable terminal wealth distribution than the constant
proportion strategy with a higher expected terminal wealth and significantly less downside risk. In addition, the optimal adaptive strategy has consistently higher expected wealth compared to the constant proportion strategy over the entire investment period. Finally, the optimal adaptive strategy is robust in the sense that it performs well on bootstrapped market data with different distributions.
\end{itemize}

\section{Formulation for Outperforming a Stochastic Benchmark}

Let the initial time $t_0=0$ and consider a set $\mathcal{T}$ of rebalancing times
\begin{equation}
\mathcal{T}\equiv \{ t_0=0 < t_1<\ldots< t_N =T\}.
\end{equation}
The fraction of total wealth allocated to each asset is adjusted  at times $t_n$, $n=0,\ldots, N-1$, with the investment
horizon $t_N = T$.
Consider an investment problem in $M$ risky and riskless assets.

Assume that, at time $t$, a fund holds wealth of amount $W_m(t)$ in asset $m$,  $m=1,\ldots,M$.
The total value of the portfolio at $t$ is then
\begin{eqnarray}
   W(t) = \sum_{m=1}^{M} W_m(t)~. \label{W_def}
\end{eqnarray}

For any given time $t$ and arbitrary function $f(t)$, define $f(t^+)=\displaystyle\lim_{\epsilon\rightarrow0^+}f(t+\epsilon)$, and $f(t^-)=\displaystyle\lim_{\epsilon\rightarrow0^+}f(t-\epsilon)$. Assume that $W(t_0^-)=0$, i.e., the initial value of the portfolio before any cash injection is zero, and let $q(t_n)$ represent an {\em a priori} specified cash injection schedule.

We denote the allocation at  stage $n$ by an
allocation vector $p_n,\; n=0,\ldots,N-1$. Given the allocation control vectors  $p_0, p_1, \ldots, p_{N-1}$,
the statistical properties of the terminal wealth of the adaptive portfolio $W(T)$  can be
determined.  Similarly, given a benchmark allocation vector ${\Tilde{p}_n}$,
the final wealth of the benchmark portfolio $W_b(T)$ can also be determined.
The time evolution of $W(t)$ and $W_b(t)$ is given by
\begin{align*}
for\;&n=0, 1, ..., N-1\\
&W(t_n^+)=W(t_n^-)+q(t_n)\\
&W_{b}(t_n^+)=W_{b}(t_n^-)+q(t_n)\\
&W(t_{n+1}^-)=p_n^TR(t_n)W(t_n^+)\\
&W_{b}(t_{n+1}^-)={\Tilde{p}_n}^TR(t_n)W_{b}(t_n^+)\\
end,&
\end{align*}
where $R(t_n)$  is the vector of returns on assets in
$(t_n^-,t_n^+)$.

Our first goal is to minimize  some measure of underperformance against the benchmark.
A natural choice  is
to quadratically penalize the underperformance of the terminal wealth of the adaptive strategy compared to a benchmark of the terminal wealth of the constant proportion strategy,
as in  \citet{li2019data}.  Note, however, that in our
case, the benchmark is stochastic. This leads to the following optimization problem ($\mathop{\mathbb{E}}[\cdot]$ is the expectation operator):
\begin{align}\label{trivial_obj}
\min_{p_0,p_1,\ldots, p_{N-1}}\mathop{\mathbb{E}}\Big[\min\big(W(T)-W_{b}(T),0\big)^2\Big]\;.
\end{align}

Unfortunately the optimal solution to \eqref{trivial_obj} is trivially the benchmark
strategy $ p_n = \Tilde{p}_n, \forall n$, which indicates the formulation \eqref{trivial_obj}
does not sufficiently capture properties of the desired solution.

We propose to generate a more ambitious
strategy  by using an  elevated target $e^{sT}\cdot W_{b}(T)$ in the objective function, i.e.,
\begin{align}
\min_{p_0,p_1,\ldots, p_{N-1}}\mathop{\mathbb{E}}\Big[\min\big(W(T)-e^{sT}\cdot W_{b}(T),0\big)^2\Big],
\label{obj_elevated}
\end{align}
where $s$ is the yearly pre-determined target outperformance spread.
Consequently, in an ideal case, the adaptive strategy will have a
terminal wealth of $e^{sT}\cdot W_{b}(T)$ which indicates that the adaptive strategy
achieves an annual outperformance spread of return $s$ compared to the benchmark strategy.

We note, however, that if the outperformance spread $s$ is large,  \eqref{obj_elevated} will tend to generate a strategy that concentrates on the asset with the highest rate of returns, which can potentially lead to an unacceptable probability
of underperformance.  We can see that outperforming a target is a complex
distribution shaping problem with multiple criteria which is difficult to formulate and to compute.
Recognizing that
low probability underperformance scenarios  often come with low probability high outperformance scenarios,
we choose an asymmetric objective function which controls the loss-side tail by penalizing the underperformance
quadratically, while at the same time penalizing the outperformance
linearly.

Our asymmetric distribution shaping benchmark outperforming formulation becomes
\begin{align}\label{final_obj}
\min_{p_0,p_1,..., p_{N-1}}\mathop{\mathbb{E}}
                   \Big[
                     \min
                         \big(
                           W(T)-e^{sT} \cdot W_{b}(T), 0
                          \big)^2
                            +
                           \max\big(
                               W(T)-e^{sT}\cdot W_{b}(T),0
                          \big)
                   \Big]  ~.
\end{align}
Figure \ref{fig:obj_function} illustrates this asymmetric distribution shaping objective function.

\begin{figure}[H]
\centering
\includegraphics[width=.48\textwidth]{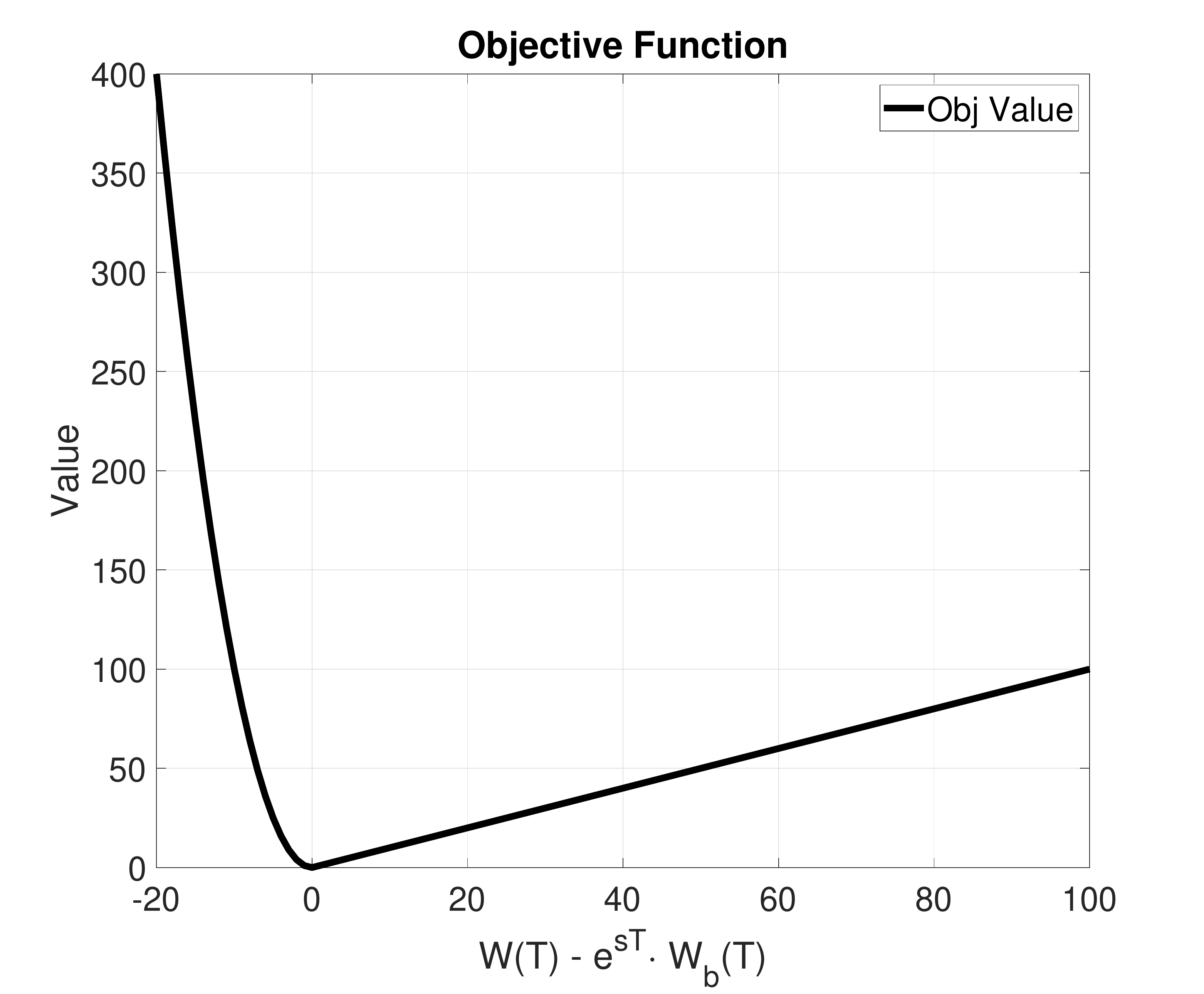}
\caption{Asymmetric distribution shaping objective function with elevated target $e^{sT}\cdot W_{b}(T)$.}
\label{fig:obj_function}
\end{figure}

We remark that distribution shaping objectives can be problem dependent and we choose
the objective function (\ref{final_obj})  for the pension investment problem. Furthermore,
our proposed framework  does not depend on any specific form of the objective function.

If we postulate parametric stochastic processes for prices of the traded assets, mathematically,
the controls $p_0,...,p_{N-1}$ can be determined using dynamic programming. This
will result in a nonlinear HJB PDE (see \citep{alaradi_2018} for example). In the
absence of any closed-form solution, computing a solution of this problem numerically would
be costly, particularly when the problem has a high dimension. Consider the simplest allocation problem,  for which
the portfolio consists of a stock index and a bond index.
In the case of  discrete rebalancing, the state variables would
be the dollar amounts in the bond and stock indices, for both the adaptive and
target portfolios \citep{dang-forsyth:2014a}. Consequently, even for this
comparatively simple case, this would result in a four-dimensional HJB  PDE.

Assume that samples of asset returns are available.
These samples can come directly from market observations or
from simulations of postulated parametric models.
Instead of solving $p_0,\ldots,p_{N-1}$ using dynamic programming,
we propose a data driven approach as follows.  We represent the optimal control as a function of several features $F(t)$,  i.e., at $t_n$, $n=0,1,\ldots,N-1$,
$$p_n=p(F(t_n))$$

\begin{example}[Two Asset Problem with Benchmark $W_{50/50}$] \label{example_1}
In our numerical examples, we will focus on portfolios consisting of two assets: a stock index and
a bond index.  The benchmark portfolio in this case will be a constant proportion strategy,
with $50\%$ stocks and $50\%$ bonds.  We will denote the wealth of the benchmark strategy in this case
as $W_{50/50}(t)$.
For this example, for the stochastic target pension allocation problem, we use three features for $F(t)$: (i) $W(t_n)$, the wealth of
the adaptive portfolio at $t_n$, (ii) $W_{50/50}(t_n)$, the wealth of
the constant proportion portfolio at $t_n$, (iii) T-t, time remaining in the investment period.
In the case that simple stochastic processes are assumed,
then it can be shown (in the absence of transaction
costs) that the controls are only a function of these features \citep{dang-forsyth:2014a} .
\end{example}

We remark that our feature set $F(t)$ for Example \ref{example_1}
is different from the features  in  \citet{samo_2016} which  explicitly use security prices.
Instead,   at time $t$ our feature set consists of the
accumulated wealth at $t$ from allocation strategy and benchmark strategy,
which depend on the returns  of traded assets from prior periods.
Traded asset  prices are not directly used  as features  for  the neural network model.
This is essentially because, at each rebalancing time, we search for the optimal
adaptive strategy amongst all strategies with the current level of wealth.
In addition,  since we evaluate the performance of a trading strategy based on the terminal wealth $W(T)$
only, the trading decision at time $t$ depends on the current accumulated wealth and return distribution of future trading periods.
Unless the asset price has predictability in its future return, including the prices as features
is redundant  in this context and will likely lead to overfitting of the model.

We use a 2-layer neural network as the functional form for the
optimal control. As a result, the goal of the optimization problem is to find the optimal
parameters of the neural network.

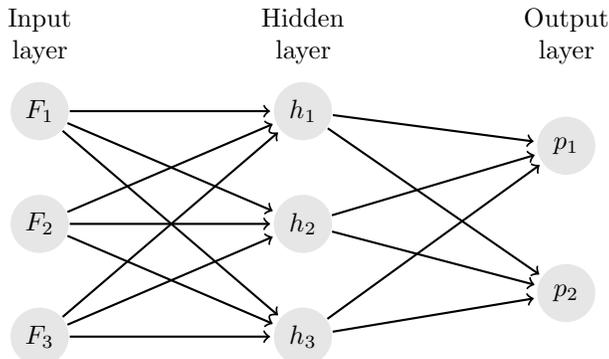
\begin{figure}[H]
\centering
\begin{tikzpicture}[shorten >=1pt,->,draw=black!100, node distance=\layersep, thick]
\tikzstyle{every pin edge}=[<-,shorten <=1pt, thick]
\tikzstyle{neuron}=[circle,fill=black!25,minimum size=22pt,inner sep=0pt]
\tikzstyle{input neuron}=[neuron, fill=black!10];
\tikzstyle{output neuron}=[neuron, fill=black!10];
\tikzstyle{hidden neuron}=[neuron, fill=black!10];
\tikzstyle{annot} = [text width=4em, text centered]

\foreach \name / \y in {1,...,3}
\node[input neuron] (I-\name) at (0,-1.5*\y) {$F_{\y}$};

\foreach \name / \y in {1,...,3}
\path[yshift=0cm] node[hidden neuron] (H-\name) at (\layersep,-1.5*\y cm) {$h_{\y}$};

\foreach \name / \y in {1,...,2}
\path[yshift=0cm] node[output neuron] (O-\name) at (2*\layersep,-1.96*\y cm) {$p_{\y}$};

\foreach \source in {1,...,3}
\foreach \dest in {1,...,3}
\path (I-\source) edge (H-\dest);

\foreach \source in {1,...,3}
\foreach \dest in {1,...,2}
\path (H-\source) edge (O-\dest);

\node[annot,above of=H-1, node distance=1cm] (hl) {Hidden layer};
\node[annot,left of=hl] {Input layer};
\node[annot,right of=hl] {Output layer};
\end{tikzpicture}
\caption{A 2-Layer NN representing the control functions}
\label{fig:NN}
\end{figure}

Assume that $h\in R^l$ is the output of the hidden layer. Let the matrix $z\in R^{dl}$ be the weights from the
input features $F(t_n)\in R^d$ to the hidden nodes $h$. We use the sigmoid activation function,
$$\sigma(u)=\frac{1}{1+e^u}\;,$$
and have
$$h_j(F(t_n))=\sigma(F_i(t_n)z_{ij})\;.$$
Here we use double summation convention, i.e.
$$F_i(t_n)z_{ij}\equiv\sum_{i=1}^{d}F_i(t_n)z_{ij},\; j = 1, ..., l\;.$$
At the output layer, we use the logistic sigmoid function as the activation function. Let the matrix $x\in R^{lM}$ be the weights for output layer. For the $m^{th}$ asset, the asset allocation on this asset is given by:
$$\bigl({p}(F(t_n))\bigr)_m=\frac{e^{x_{km}h_k(F(t_n))}}{\sum_ie^{x_{ki}h_k(F(t_n))}},\; 1\leq m\leq M\;.$$
Note that with the logistic sigmoid activation function, the following constraint is automatically satisfied
$$0\leq p(F(t_n))\leq1,\;1^T{p}(F(t_n))=1\;.$$
This enforces the constraints of no-shorting and no leverage. In addition, insolvency cannot occur.

The dynamics of the terminal wealth of the adaptive portfolio then becomes
\begin{align}\label{one_path}
for\;&n=0, 1, ..., N-1 \nonumber \\
&W(t_n^+)=W(t_n^-)+q(t_n) \nonumber\\
&W(t_{n+1}^-)=p(F(t_n))^TR(t_n)W(t_n^+) \nonumber \\
end&\; ~.
\end{align}

We approximate the expectation in equation (\ref{final_obj}) by a finite number of
wealth samples of $W(T)$, computed from return samples of $R(t_n)$ obtained by bootstrapping the historical data.
Let $W^{\ell}(T), W_{b}^{\ell}(T)$ be the final wealth samples for the adaptive
and benchmark strategies, obtained using equation (\ref{one_path}), along the
$\ell^{th}$ return sample path $R(t_n)^\ell,~n=0,1,\ldots, N-1$ .

Denote
\begin{eqnarray}
g( x ) \equiv
\min
\big(
x, 0
\big)^2
+
\max\big(
x,0
\big)
~.
\end{eqnarray}
The expectation in equation (\ref{final_obj}) is approximated by
\begin{eqnarray} \label{obj1}
\mathop{\mathbb{E}}
\Big[
g( W(T) - e^{sT} \cdot W_{b}(T) )
\Big] \simeq \frac{1}{L} \sum_{\ell=1}^{\ell = L} {g} (W^{\ell}(T) - e^{sT} \cdot W^{\ell}_{b}(T) )
\end{eqnarray}
Since the approximate function on the right hand side of \eqref{obj1}
is a nonconvex, continuous but piecewise differentiable function of the NN weights, solving the optimization problem is challenging.

We recognize however that $\mathbb{E}
\Big[
g( W(T) - e^{sT} \cdot W_{b}(T) )\Big]$ is a continuously
differentiable function of the NN weights assuming that the continuous return distribution is continuous.
This motivates us to use the
smoothing
technique from \citet{alexander2006minimizing}.
In equation (\ref{obj1}), we replace $g(x)$ by the smoothed approximation $\bar{g}(x)$, to obtain a continuously differentiable approximation,
\begin{equation}
\bar{g}(x)=\begin{cases}
x, & \text{if $x>\epsilon$}\;,\\
\frac{x^2}{4\epsilon}+\frac{1}{2}x+\frac{1}{4}\epsilon, & \text{if $-\epsilon\leq x\leq\epsilon$}\;,\\
(x+\epsilon)^2,&\text{if $x<-\epsilon$}\;,
\end{cases}
\end{equation}
where $\epsilon$ is a predetermined small number.
Since we are essentially optimizing the parameters $x$ and $z$, we write the final problem as
\begin{equation}\label{approx_obj}
\min_{x,z}~ \frac{1}{L} \sum_{\ell = 1}^{\ell = L} \bar{g} (W^{\ell}(T) - e^{sT} \cdot W_{b}^{\ell}(T) ) \;.
\end{equation}
Similar to \citet{li2019data}, we use the same trust region optimization method \cite{coleman1996} to solve the resulting optimization problem.

More specifically, the optimization method requires the evaluation of the objective function, its derivative with respect to the weight parameters $x$ and $z$, and the Hessian matrix. The gradients can be explicitly evaluated via the chain rule, and the Hessian matrix can be numerically computed via the finite-difference of the gradients. The detailed gradient computation can be found in \citet{li2019data}.

\section{Testing on Different Distributions with Bootstrap Resampling}\label{append:proof}

Success in data-driven learning critically depends on the efficient use of data.
Standard machine learning measures success based on testing the model performance on
unseen data which are assumed to have the same distribution as the training data. In other words,
test results are typically computed based on test samples from the same distributions as training samples.

For training of the optimization problem \eqref{approx_obj},
we only have access to a single path of historical returns. This lack of data presents a unique challenge in data-driven financial model learning.

For financial model learning and testing, it is a common practice to train and test strategy performance by splitting the historical market data path into two segments - one for training and the other for testing. A critical problem in this approach is insufficient data for robust learning and testing. This is especially problematic in the context of pension planning due to the long-term investment horizon.

\citet{li2019data} uses block bootstrap resampling to generate training and testing data in data-driven financial decision learning.
Standard block bootstrap resampling is done by dividing the historical market sequential data into blocks with fixed blocksizes and randomly choosing blocks to construct the bootstrap resampled data series. To reduce the impact of a fixed blocksize and to mitigate the edge effects at each block end, the stationary block bootstrap \citep{patton2009correction,politis2004automatic} can be used. A single bootstrap resampled path is constructed as follows.

\begin{itemize}
\item First, randomly select a block of the historical market data time series. The blocksize is randomly sampled from a shifted geometric distribution with an expected blocksize $\hat{b}$. The optimal choice for $\hat{b}$ is determined using the algorithm described in \citep{patton2009correction}.
\item Repeat the previous step and concatenate the new block after the existing data series until the new resampled path has reached the
desired length.
\item If the selected block exceeds the range of historical data, wrap around the historical data as in the circular bootstrap method \citep{politis2004automatic,patton2009correction}.
\end{itemize}

Algorithm \ref{Algo:bootstrap} presents pseudocode for the stationary block bootstrap.

\begin{algorithm}[htp]
\SetAlgoLined
\tcc{initialization}
bootstrap\_samples = [ ]\;
\tcc{loop until the total number of required samples are reached}
\While{True }{
\tcc{choose random starting index in [1,\ldots,N], N is the index of the last historical sample}
index = UniformRandom( 1, N )\;
\tcc{actual blocksize follows a shifted geometric distribution with expected value of exp\_block\_size}
blocksize = GeometricRandom( $\frac{1}{exp\_block\_size}$ )\;
\For{$i = 0;\ i < blocksize;\ i = i + 1$}{
\tcc{if the chosen block exceeds the range of the historical data array, do a circular bootstrap}
\eIf{index + i $>$ N}{
bootstrap\_samples.append( historical\_data[ index + i - N ] )\;
}{bootstrap\_samples.append( historical\_data[ index + i ] )\;}

\If{bootstrap\_samples.len() == number\_required}{\Return bootstrap\_samples\;} }
}
\caption{Pseudocode for stationary block bootstrap}
\label{Algo:bootstrap}
\end{algorithm}

In \citet{li2019data}, the training dataset is generated using stationary block resampling with one expected blocksize and the testing dataset is generated with a different expected blocksize. As \cite{politis1994stationary} points out, changing the expected blocksizes for block bootstrap resampling essentially changes the distribution of the bootstrap resampled data paths. Consequently, such training and testing assessments actually perform out-of-distribution tests.

Intuitively, using the block bootstrap resampling time-series financial market data seems natural. We have trained a model, considering all permutations of the financial market data with respect to different and random concatenations of time horizons. In addition, testing has been performed on a different distribution of the financial market random horizon concatenations, since the testing data uses a different expected blocksize from that of the training data. Indeed, evaluating testing performance in this fashion seems to uphold a more stringent standard in comparison to the standard machine learning which evaluates testing performance assuming (unseen) testing samples are from the same distribution of the training data.

Still, one may have concerns that when the training data and testing data are block bootstrap resampled from the same underlying historical market data sequence, one path may appear in both training and testing datasets so that the learning algorithm may benefit from such an unfair edge.
To address such concerns, we establish a theoretical bound on the probability of training and testing sample sequences being exactly the same.

\begin{theorem} \label{thm:fix}
Consider generating a sequence of $N$ data points using fixed block resampling from a sequence of $N_{tot}$ distinct observations.
Let path $\ppath_1$ be a bootstrap resampled with a fixed blocksize of $b_1$ and path $\ppath_2$ be a bootstrap resampled with a fixed blocksize of $b_2$. Then the probability of $\ppath_1$ and $\ppath_2$ being identical is $(\frac{1}{N_{tot}})^{lcm(\frac{N}{b_1},\frac{N}{b_2})}$, where $lcm(a,b)$ is the least common multiple of integer $a,b$.
\end{theorem}

The proof of Theorem \ref{thm:fix} is in Appendix \ref{appendix:proof}.
To put this into perspective, assume a fixed blocksize for the training paths of 6 months, and
a fixed blocksize for the testing path of 24
months (or 2 years). Consider a 30-year investment horizon of monthly return paths randomly generated from historical monthly data over 90 years, i.e. $N=30\times12=360$ and $N_{tot}=90\times12=1080$. Then the probability of a training path being identical to a testing path is $(\frac{1}{1080})^{lcm(\frac{360}{6},\frac{360}{24})}=(\frac{1}{1080})^{60}<10^{-180}.$
Assume that we use a total of 100,000 training paths in the training data and 10,000 testing paths in the testing data.
By the union bound, the probability of the existence of a
pair of identical training and testing paths is bounded by $100,000\times10,000\times 10^{-180}=10^{-171}.$

Next, we consider the stationary block bootstrap case, in which the blocksizes are randomly generated from a
shifted geometric distribution. We are able to establish the following
theorem about the probability of two paths generated with stationary block bootstrap being identical.

\begin{theorem}\label{thm:stb}
Consider generating a sequence of $N$ data points using stationary block resampling from a sequence of $N_{tot}$ distinct observations.
Let $\ppath_1$ and $\ppath_2$ be two paths generated from the stationary block bootstrap resampling from this observation sequence with the expected blocksizes of $\hat{b}_1$ and $\hat{b}_2$ respectively, and both have a length of $N$. The probability of $\ppath_1$ and $\ppath_2$ being identical is

$$\frac{1}{N_{tot}}\Big(\big(1-\frac{1}{\hat{b}_1}\big)\big(1-\frac{1}{\hat{b}_2}\big)+\frac{\frac{1}{\hat{b}_1}+\frac{1}{\hat{b}_1}-\frac{1}{\hat{b}_1\hat{b}_2}}{N_{tot}}\Big)^{N-1}.$$

\end{theorem}

The proof of  Theorem \ref{thm:stb} is  also in Appendix \ref{appendix:proof}.
Consider the following example.   If  the training paths are bootstrap resampled with an
expected blocksize of 6 months (0.5 years) and the testing paths with an expected
blocksize of 24 (2 years), for  $N=30\times12=360$ (30-year horizon) and
$N_{tot}=90\times12=1080$. Then the probability of a training path being identical to a testing path is $8.737\times10^{-39}.$

If training data set consists of a total of 100,000 training paths and  testing data set consists of 10,000 testing paths, by union bound, the probability of existing a pair of training and testing path being identical is bounded by $100,000\times10,000\times 8.737\times10^{-39}<10^{-29}.$

Therefore, even when the training set and testing set are generated from the same data sequence, the probability of observing the same path in the training and testing dataset is near zero. This suggests that using the block bootstrap resampling to generate training and testing data sets is a robust method for enhancing data for the learning framework.

\begin{remark}
Under stationary block bootstrap, a path is likely to have large actual blocksizes even if the expected blocksize is relatively small, which can result in a higher probability of observing two identical paths than under fixed block bootstrap. For example, a path with expected blocksize of $10$ years has a $5\%$ probability of only containing one block of $30$ years, which increases the probability of one path being identical to another path, according to Theorem \ref{Lemma:LemmaA1}.
\end{remark}

\section{Performance Assessment and Comparison}
We evaluate  and report the performance of the proposed data-driven approach for outperforming a stochastic target in the context of a 30 year DC pension plan.
In our numerical tests, we focus on portfolios with only two assets: a stock index and a bond index, as
described in Example \ref{example_1}. The benchmark portfolio is a constant weight strategy, which is rebalanced
to $50\%$ bonds and $50\%$ stocks annually.  We denote the wealth of the benchmark strategy at time $t$  by $W_{50/50}(t)$.

\subsection{Original Data and Its Augmentation}

\subsubsection{Historical Data}
Our main objective here is to consider the core allocation problem
between a risky and a defensive asset.
To that end, we use monthly historical data from the Center for Research in i
Security Prices (CRSP) from January 1, 1926  to December 31, 2015.\footnote{More specifically,
results presented here were calculated based on data from Historical Indexes,
\copyright 2015 Center for Research
in Security Prices (CRSP), The University of Chicago Booth School of Business.
Wharton Research Data Services was used in preparing this article.
This service and the data available thereon constitute valuable
intellectual property and trade secrets of WRDS and/or its third-party suppliers.}. Specifically, we use the CRSP 3-month Treasury bill (T-bill) index
and the CRSP cap-weighted total return index. The latter index includes all distributions
for all domestic stocks trading on major U.S. exchanges. Since both indexes
are in nominal terms, we adjust them for inflation  using the
U.S. CPI index, also supplied by CRSP. We use real indexes since
investors saving for retirement should be
focused on real (not nominal) wealth goals.
Note that in \citep{li2019data}, in the context of a fixed (non-stochastic) target based
objective function, we have also tested the use of the CRSP equal weighted index (for the
risky asset) and the ten year treasury index (for the defensive asset).  The control strategies
are qualitatively similar for either choice of risky and defensive asset.  For simplicity here,
we will focus on the CRSP index and the 3-month T-bill case.

For illustration, we consider here a two-asset allocation  in which the wealth of the portfolio is allocated to the two indexes. We subsequently refer to the two assets simply as the stock and the bond.

For the stock index and bond index,  Table \ref{tb:blocksize} shows the optimal expected blocksize for each index estimated from the historical data. When using the resampling method in the proposed data-driven NN approach, we simultaneously sample the same block across all asset data sets (i.e. the stock index and bond index).  Since  the optimal blocksize varies with the index, it is not clear which  blocksize to use  since we need to simultaneously resample both indices.
Consequently, we will carry out tests with a variety of blocksizes, in the ranges reported in Table \ref{tb:blocksize}.

\begin{table}[htp]
\centering
\begin{tabular}{lc}
\hline
Data Series & \begin{tabular}[c]{@{}c@{}}Optimal expected\\ block size $\hat{b }$ (months)\end{tabular} \\ \hline
Real 3-month T-bill index & 50.1 \\ \hline
Real CRSP cap-weighted index & 1.8 \\ \hline
\end{tabular}
\caption{Optimal expected blocksize $\hat{b}= 1/v$ when the blocksize follows a geometric distribution $Pr(b = k) =
(1 - v)^{k-1}v$. The algorithm in \citet{patton2009correction} is used to determine $\hat{b}$.}
\label{tb:blocksize}
\end{table}

\subsection{Experiment Setting}
The parameters used in training and testing the proposed data-driven approach are as below:
\begin{itemize}
\item $L $: a total of $L=100,000 $ bootstrap paths are used for training;
\item $L_{test} $: a total of $L_{test} =10,000 $ paths are bootstrap resampled from a different expected blocksize than the training data for testing the strategy performance;
\item $W(0) $: initial wealth is $W(0) =0$;
\item $T $: the entire investment period is $T = 30$ years;
\item $N $: the entire period is divided into $N =30$ periods. At the beginning of each period rebalancing occurs, i.e., annual rebalancing;
\item $q $: annual cash injection is $q =10$;
\item $s $: the annual target outperformance rate $s = 1\%$ for calculating the elevated target $e^{sT}W_{50/50}(T)$, where $W_{50/50}(T)$ is the terminal wealth of the constant proportion portfolio;
\item 3 features:
\begin{itemize}
\item $T-t$: time remaining in the investment period,
\item $W(t)$: wealth of the adaptive portfolio at time $t$,
\item $W_{50/50}(t)$: wealth of the constant proportion portfolio at time $t$.
\end{itemize}
\end{itemize}

\subsection{Performance}\label{sec:mkt}
We now evaluate the performance of the optimal adaptive strategy trained on bootstrap resampled data. First, we show the performance of the optimal adaptive strategy trained on the bootstrap resampled data with the expected blocksize $\hat{b}$ = 0.5 years,
and tested on bootstrap resampled data with expected blocksize of $\hat{b}=2$, which is the average optimal blocksize.
When discussing robustness in Section \ref{sec:robustness_blocksize},
we show that the strategy performance using alternative training-testing expected blocksize pairs is qualitatively similar.

\begin{table}[H]
{\scriptsize
\begin{center}
\resizebox{\columnwidth}{!}{
\begin{tabular}{lccccc} \hline
\multicolumn{6}{c}{ Training Results on Bootstrap Data: Expected Blocksize $\hat{b}=0.5$ years} \\ \hline
Strategy & $E(W_T)$ & $std(W_T)$& $median(W_T)$ & $Pr(W_T<median(W_T^{CP}))$ & $Pr(W_T<median(W_T^{NN}))$ \\ \hline
constant proportion($p=0.5$) & 678 & 276 & 624 & 0.50 & 0.84 \\
adaptive & 963 & 474 & 913 & 0.27 & 0.50 \\\hline
\multicolumn{6}{c}{ Testing Results on Bootstrap Data: Expected Blocksize $\hat{b}=2$ years} \\ \hline
Strategy & $E(W_T)$ & $std(W_T)$& $median(W_T)$ & $Pr(W_T<median(W_T^{CP}))$ & $Pr(W_T<median(W_T^{NN}))$ \\ \hline
constant proportion($p=0.5$) & 679 & 267 & 629 & 0.50 & 0.84 \\
adaptive & 962 & 449 & 921 & 0.26 & 0.50 \\\hline
\end{tabular}
}
\end{center}
}
\caption{Terminal wealth statistics of the optimal adaptive strategy, trained on bootstrap resampled data with blocksize $\hat{b}=0.5$ years and tested on bootstrap resampled data with blocksize $\hat{b}=2$ years.}
\label{tb:mkt_terminal}
\end{table}

Table \ref{tb:mkt_terminal} reports performance statistics and the probability of the terminal wealth less than the median of the terminal wealth of both strategies.
From Table \ref{tb:mkt_terminal} , we observe that
\begin{itemize}
\item The median and mean of the optimal adaptive strategy is significantly higher than the constant proportion strategy.
\item The optimal adaptive strategy has only 26\% probability of achieving a lower terminal wealth than the median terminal wealth of the constant proportion strategy ($median(W_T^{CP})$), while the constant proportion strategy has an 84\% probability of achieving a lower terminal wealth than the median terminal wealth of the NN adaptive strategy ($median(W_T^{NN})$).
\end{itemize}

It is also worth noting that the standard deviation of the terminal wealth of the optimal adaptive strategy is higher than the standard deviation of the terminal wealth of the constant proportion strategy. In the context of dynamic trading,
a higher standard deviation does not
imply that the performance of the strategy is poor.
In fact, we can observe from Figure \ref{fig:mkt_test_hist} that the distribution of the terminal wealth of the optimal adaptive strategy is significantly more right-skewed. A higher standard deviation of terminal wealth is desirable in the right-skewed situation \citep{van2019distribution}. This illustrates why standard deviation and Sharpe Ratio are poor measures of risk for inherently non-linear strategies \citep{lhabitant_2000}.
In fact, the optimal adaptive dynamic strategy has properties in common with option-based strategies. We also plot the CDF plot for the optimal adaptive strategy and the constant proportion strategy in Figure \ref{fig:mkt_test_wealth_cdf}.

\begin{figure}[htp]
\centering
\begin{subfigure}[b]{0.34\textwidth}
\centering
\includegraphics[width=\textwidth]{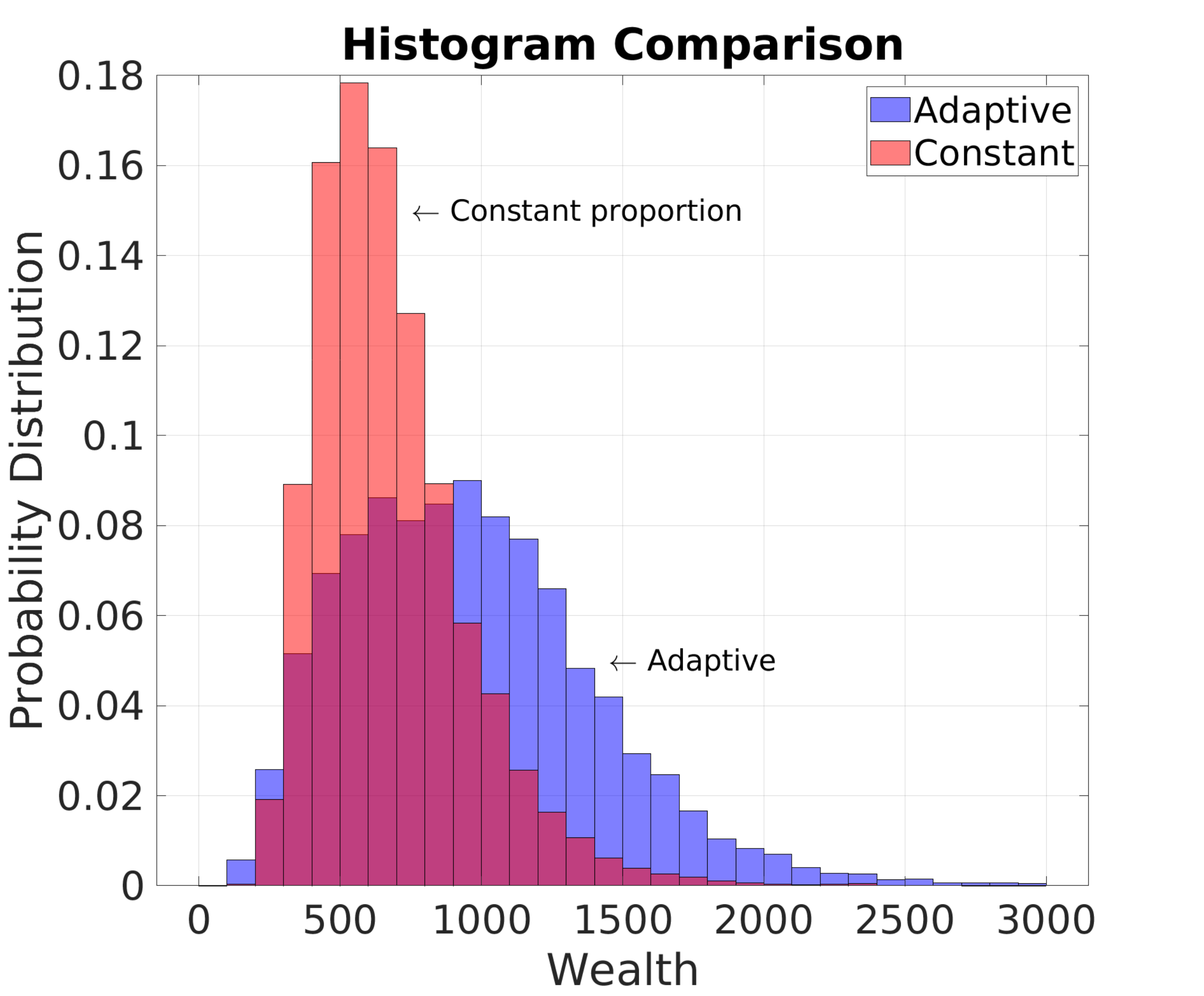}
\caption{Histogram of terminal wealth for adaptive strategy and constant proportion strategy}
\label{fig:mkt_test_hist}
\end{subfigure}
\hfill
\begin{subfigure}[b]{0.30\textwidth}
\centering
\includegraphics[width=\textwidth]{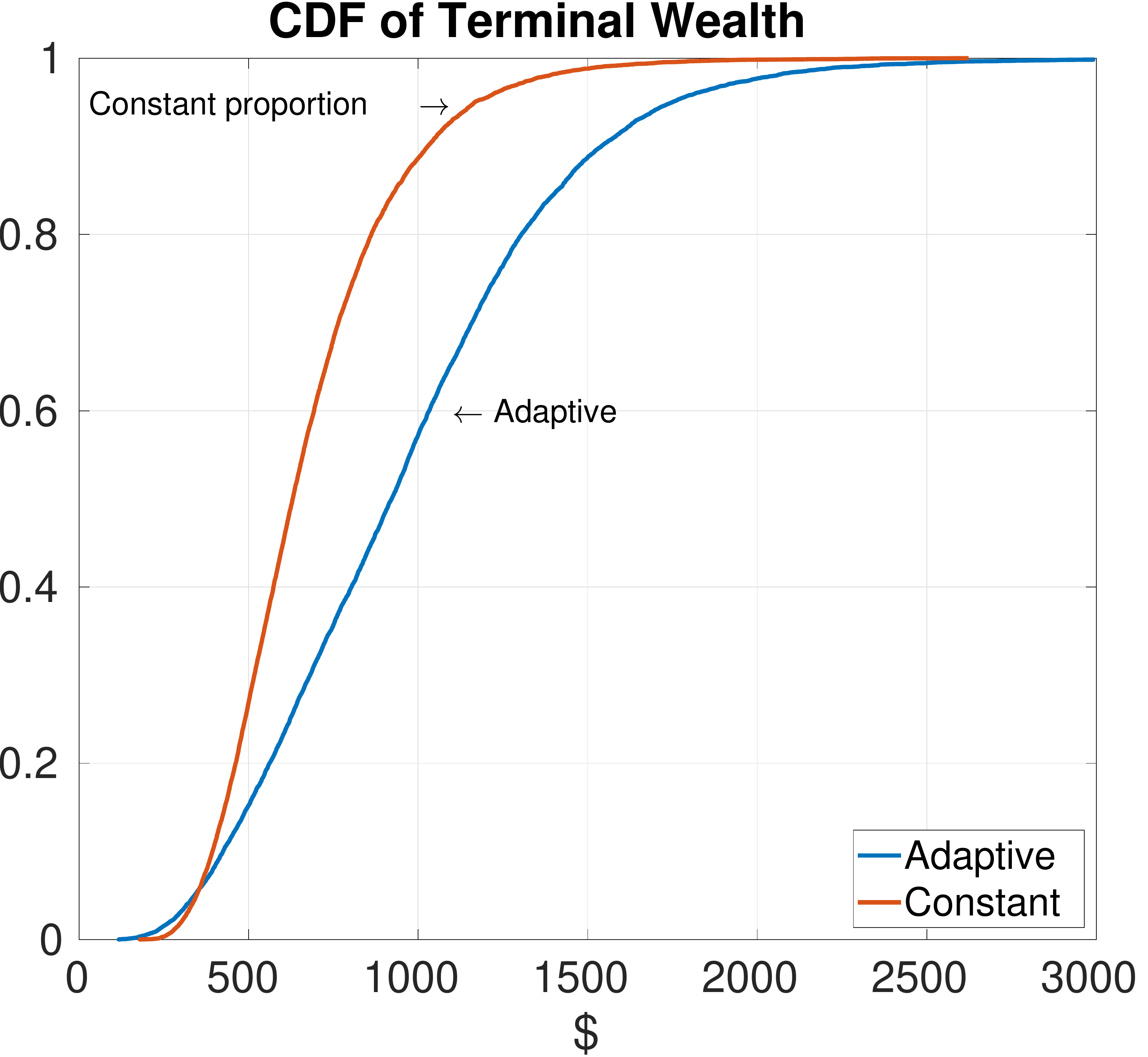}
\caption{CDF of terminal wealth adaptive strategy and constant proportion strategy}
\label{fig:mkt_test_wealth_cdf}
\end{subfigure}
\hfill
\begin{subfigure}[b]{0.31\textwidth}
\centering
\includegraphics[width=\textwidth]{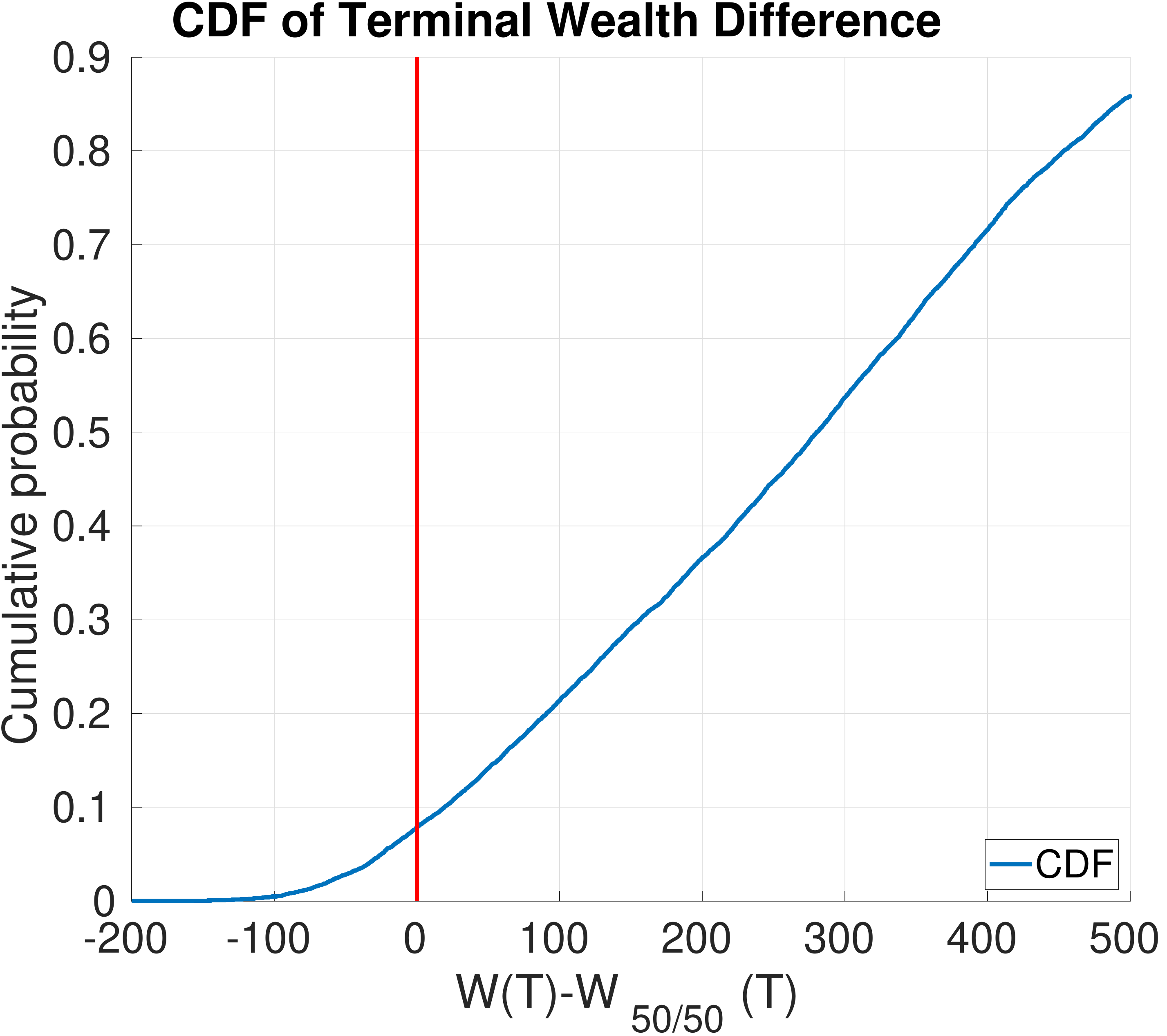}
\caption{CDF of terminal wealth difference between adaptive strategy and constant proportions strategy}
\label{fig:mkt_test_wdiff_cdf}
\end{subfigure}
\caption{Histogram of terminal wealth $W(T)$ (adaptive) and $W_{50/50}(T)$ (constant proportion) and CDF of wealth difference $W(T)-W_{50/50}(T)$ based on the testing data (bootstrap data with $\hat{b}$=2 years) }
\label{fig:mkt_hist}
\end{figure}

We should point out that the terminal wealth distribution of the optimal adaptive strategy
has a slightly worse left tail than the constant proportion strategy.
The 90\% VaR of terminal wealth is 340 for the optimal adaptive strategy and 394 for the constant proportion strategy.\footnote{We measure quantiles
of the terminal wealth, not losses. Hence a larger value of VAR is more desirable, i.e. has less risk.}
These tail events
occur when the bootstrapped path corresponds to
consistently bearish market periods when stocks underperform bonds for a long period of time.
We remark that the investor can
include risk measures such as
VaR and CVaR in the objective function if reducing the tail risk is of a higher priority in the investment plan. This, of course,
will produce a lower probability of outperformance.

Figure \ref{fig:mkt_test_wdiff_cdf} shows the cumulative distribution function (CDF) of the wealth difference $W(T) - W_{50/50}(T)$ to give a more direct comparison between the optimal adaptive strategy and the constant proportion strategy along the same paths. From Figure \ref{fig:mkt_test_wdiff_cdf} we can see that the probability of the optimal adaptive strategy underperforming the constant proportion strategy is less than 10\%. When underperformance occurs, the magnitude of underperformance is small compared to the magnitude of outperformance.

We have analyzed and compared the overall performance based on the terminal wealth adaptive strategy. Next, we provide more detailed comparisons of the various characteristics of the strategies.

\subsubsection{Strategy Performance Over Time}

Since the objective function for the optimal control \eqref{final_obj} is defined from the terminal wealth, we examine how the optimal adaptive strategy performs over the entire period of investment.

\begin{figure}[htp]
\centering
\begin{subfigure}[b]{0.48\textwidth}
\centering
\includegraphics[width=\textwidth]{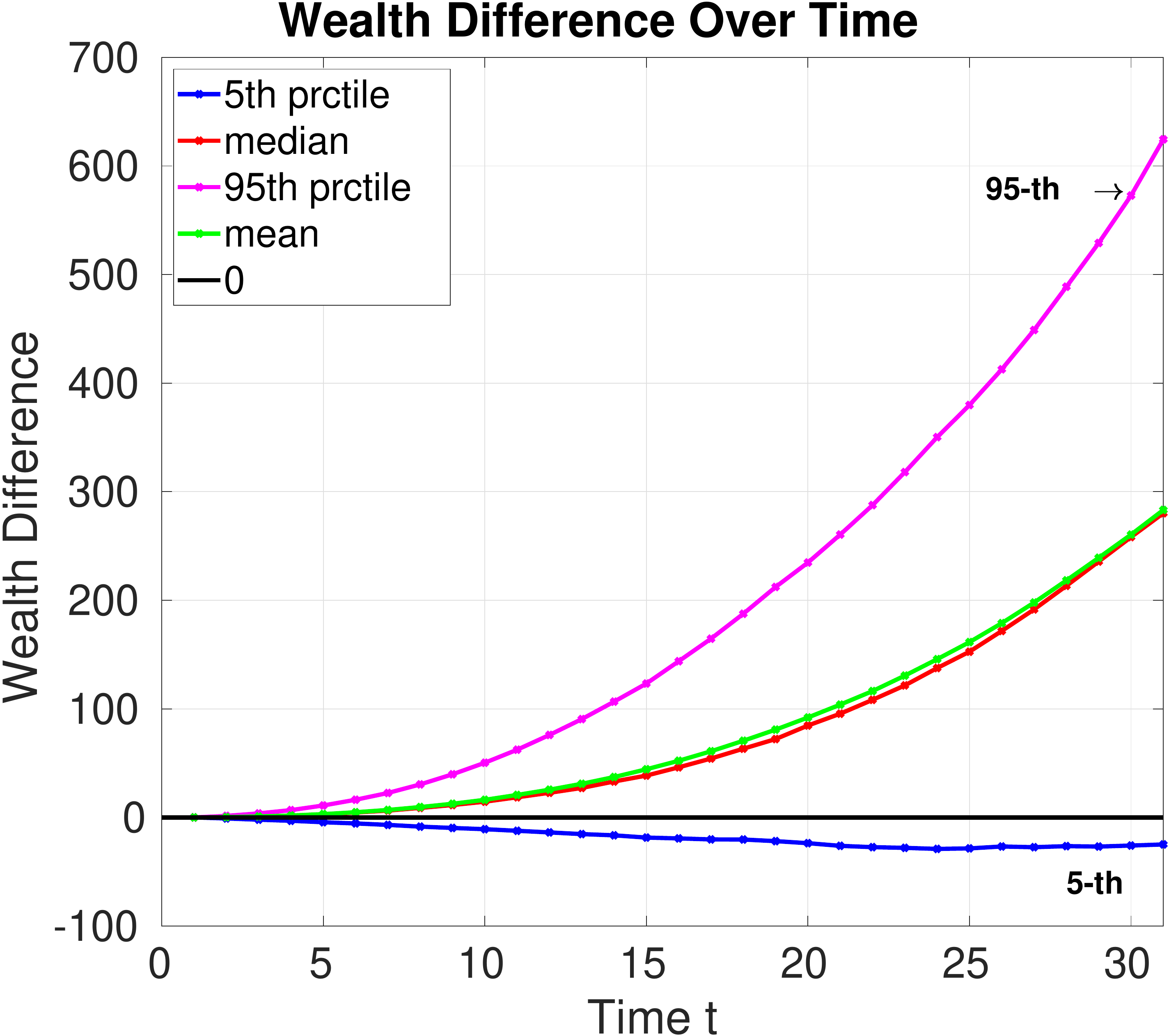}
\caption{Percentiles of wealth difference $W(t)-W_{50/50}(t)$ over time}
\label{fig:mkt_wdiff_ot}
\end{subfigure}
\hfill
\begin{subfigure}[b]{0.47\textwidth}
\centering
\includegraphics[width=\textwidth]{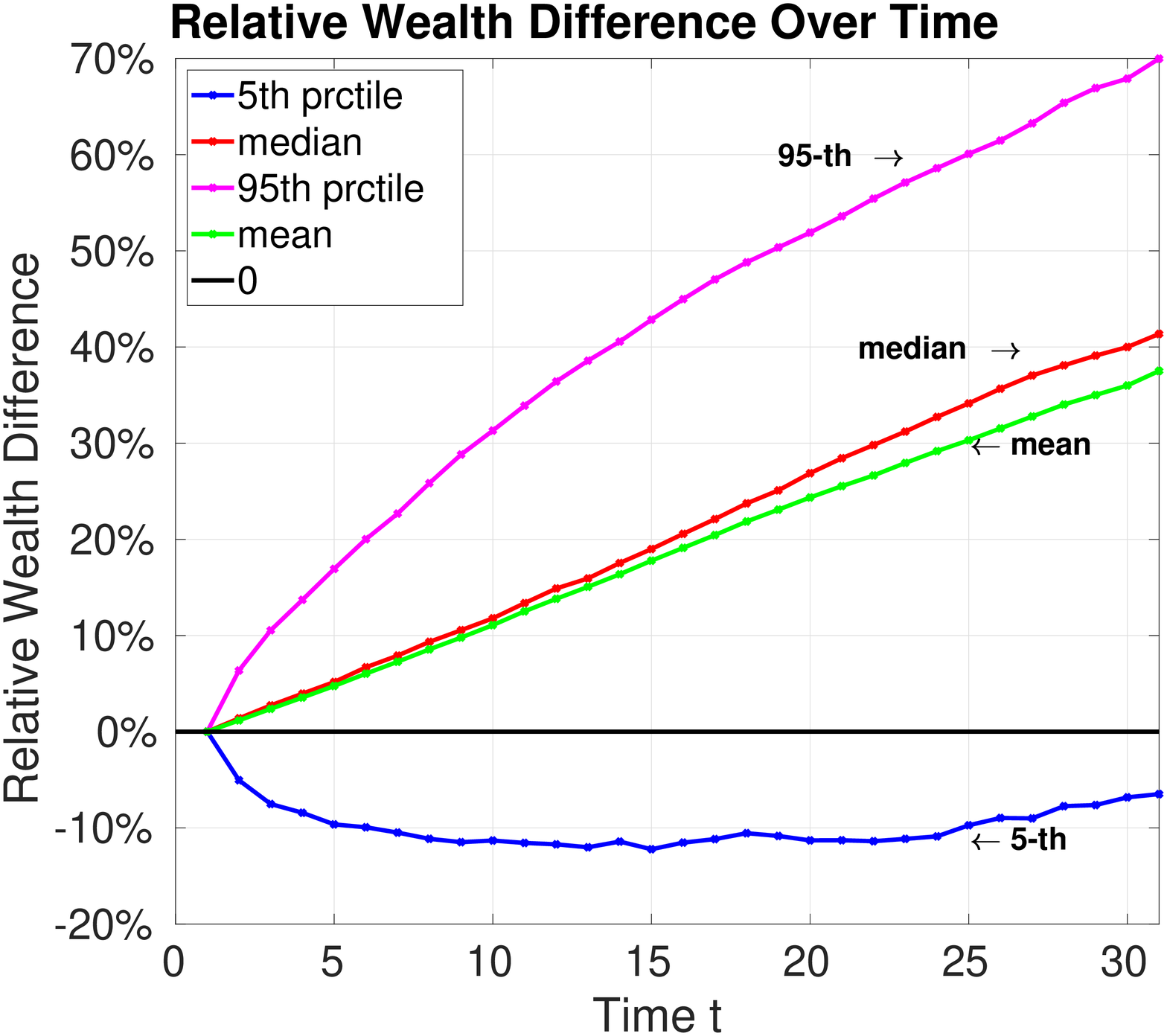}
\caption{Percentiles of relative wealth difference $\frac{W(t)-W_{50/50}(t)}{W_{50/50}(t)}$ over time}
\label{fig:mkt_rel_wdiff_ot}
\end{subfigure}
\caption{Wealth difference and relative wealth difference over time: $W(t)$ denotes the optimal adaptive is wealth and $W_{50/50}(t)$ denotes the benchmark }
\label{fig:w_ot}
\end{figure}

Figure \ref{fig:w_ot}
graphs the average and various percentiles of the wealth difference $W(t)-W_{50/50}(t)$ in the investment time horizon.
From Figure \ref{fig:w_ot}, we observe that
\begin{itemize}
\item With a high probability, the optimal adaptive strategy achieves higher wealth than the constant proportion strategy over time.
\item The outperformance of the optimal adaptive strategy in terms of the relative wealth difference is not as significant as the wealth difference in dollar values.
\end{itemize}

The observations indicate that larger outperformance of the optimal adaptive strategy often occurs when the
constant proportion strategy performs well. Nevertheless, the outperformance of
the optimal adaptive strategy in terms of the relative wealth difference is still very
impressive with a median value of almost 40\% at the terminal stage.
Of course, if we are primarily interested in relative outperformance, it is a simple
matter to alter our objective function to focus on achieving this goal.

Figure \ref{fig:w_ot} shows that, even though the objective function only
targets the wealth difference of the portfolios at the terminal time, without having any direct restrictions on the wealth of the optimal adaptive strategy in the interim period, the adaptive strategy still manages to have a statistically higher wealth throughout the entire investment period.

\subsubsection{Strategy Characteristics}
We further examine the characteristics of the optimal adaptive strategy.
Figure \ref{fig:allocation} shows different percentiles of the stock allocation of the optimal adaptive strategy over time. We observe that
\begin{itemize}
\item In general, the stock allocation (fraction of wealth invested in stocks) decreases when approaching the end of the investment horizon.
\item The stock allocation almost always stays above the benchmark allocation of 50\%.
\end{itemize}

\begin{figure}[htp]
\centering
\begin{subfigure}[b]{0.45\textwidth}
\centering
\includegraphics[width=\textwidth]{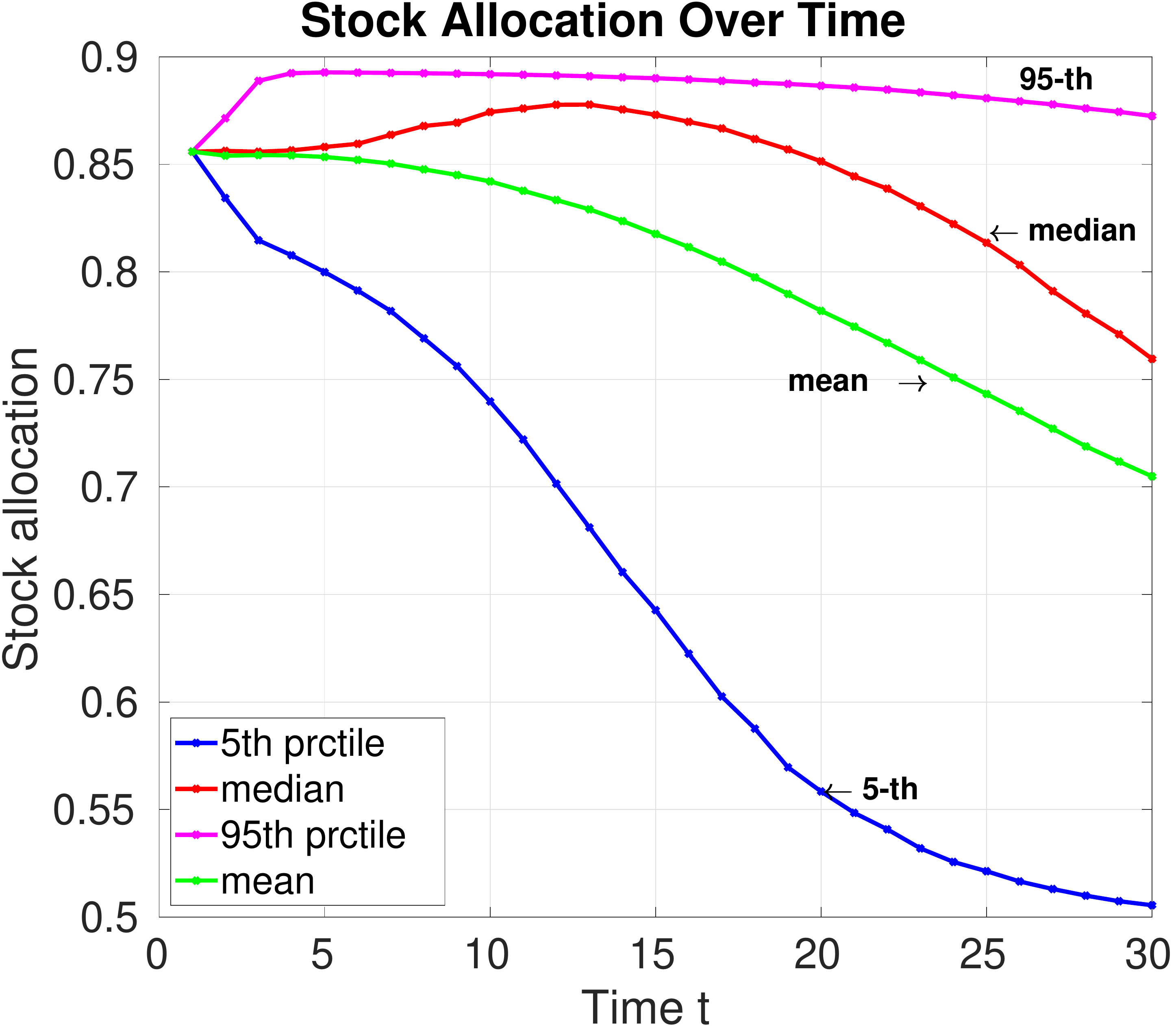}
\caption{Percentiles of the fraction invested in stocks over time for the adaptive strategy}
\label{fig:allocation}
\end{subfigure}
\hfill
\begin{subfigure}[b]{0.47\textwidth}
\centering
\includegraphics[width=\textwidth]{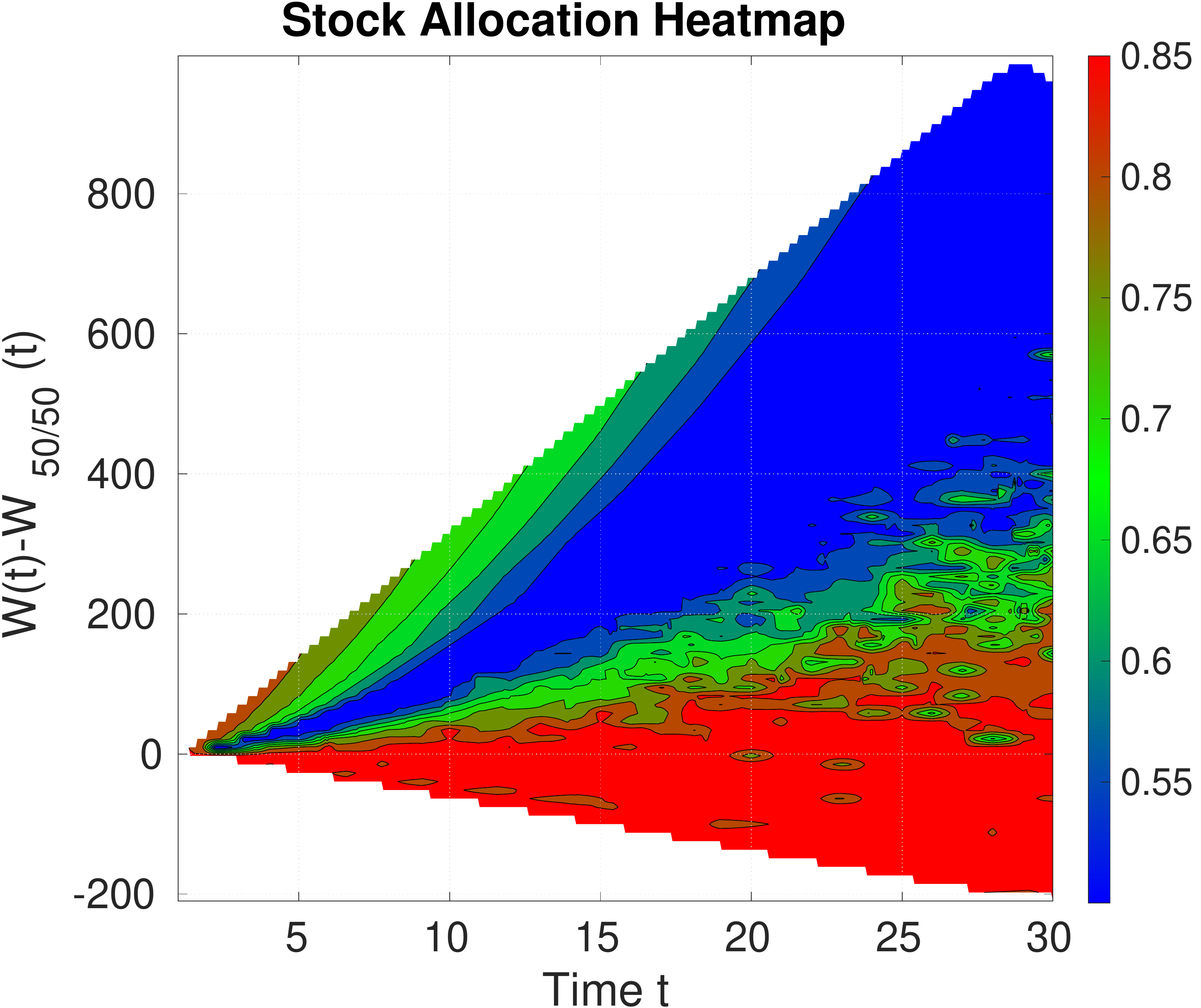}
\caption{Heatmap, fraction invested in stocks for the adaptive strategy}
\label{fig:heatmap}
\end{subfigure}
\caption{Fraction invested in stocks over time for the optimal adaptive strategy: percentiles and the heatmap}
\label{fig:strategy}
\end{figure}

With a red-blue color scheme, Figure \ref{fig:heatmap} shows the heatmap of the stock allocation with respect to time $t$ and the wealth difference $W(t)-W_{50/50}(t)$.
Darker shades of the red color indicate more allocation in stocks and darker shades of the blue color indicate more allocation in bonds.

From Figure \ref{fig:heatmap}, we observe that when $W(t)-W_{50/50}(t)$ is positive and large
(optimal adaptive strategy outperforming), the allocation of wealth to the stock becomes small.
The intuitive explanation is that the optimal adaptive strategy tends to decrease the wealth allocation to stocks
once it has established an advantage over the benchmark constant proportion strategy.
This also explains why the stock allocation almost always stays above 50\%.
In most cases where the optimal adaptive strategy has established an advantage over the
constant proportion strategy (as we have seen in Figure \ref{fig:w_ot}),
decreasing the stock allocation to 50\% to maintain the same allocation strategy as the 50/50 constant proportion strategy
locks in the outperformance.

On the other hand, when $W(t)-W_{50/50}(t)<0$ (i.e. the adaptive strategy underperforms),
the optimal policy allocates more wealth to stocks. This is because
the stock index has a higher expected return than the bond index.
To eventually outperform the constant proportion strategy, the adaptive strategy
invests more wealth in stocks, in an attempt to make up for the lost ground.

In fact, the optimal adaptive strategy appears to be a
contrarian strategy,
following which an investor buys and sells in opposition to the prevailing sentiment at the time.

\subsubsection{Historical Backtest Performance}
As a special out-of-sample test, we consider the actual historical path from 1985 to 2015 to backtest the performance of the optimal adaptive strategy. We note that the historical path is not a path in the training data set.

From Figure \ref{fig:backtest}, we see that the optimal adaptive portfolio always maintains a higher wealth than the constant proportion strategy over the entire investment period. While optimizing the performance of the adaptive strategy on a specific path is not the goal of our study, it is still quite interesting to see that historically the optimal adaptive strategy does better than the constant proportion strategy.

Note that the adaptive strategy does show a large drawdown in 2002 and 2008. However, our objective function is posed in terms of outperformance of the terminal wealth. We see that the adaptive strategy outperforms, in the sense that its wealth is always above the benchmark wealth, even in 2002 and 2008. It is, of course, possible to add penalties on drawdowns in the objective function. However, this would result in less favorable terminal statistics.

The solid line without markers in Figure \ref{fig:backtest} illustrates the time evolution of the stock allocation on the historical path. When the adaptive strategy performs poorly, such as in 2002 and 2008, the strategy allocates more wealth to stocks. When the adaptive strategy performs well, the strategy decreases allocation to stocks and invests more in bonds.

\begin{figure}[htp]
\centering
\includegraphics[width=0.48\textwidth]{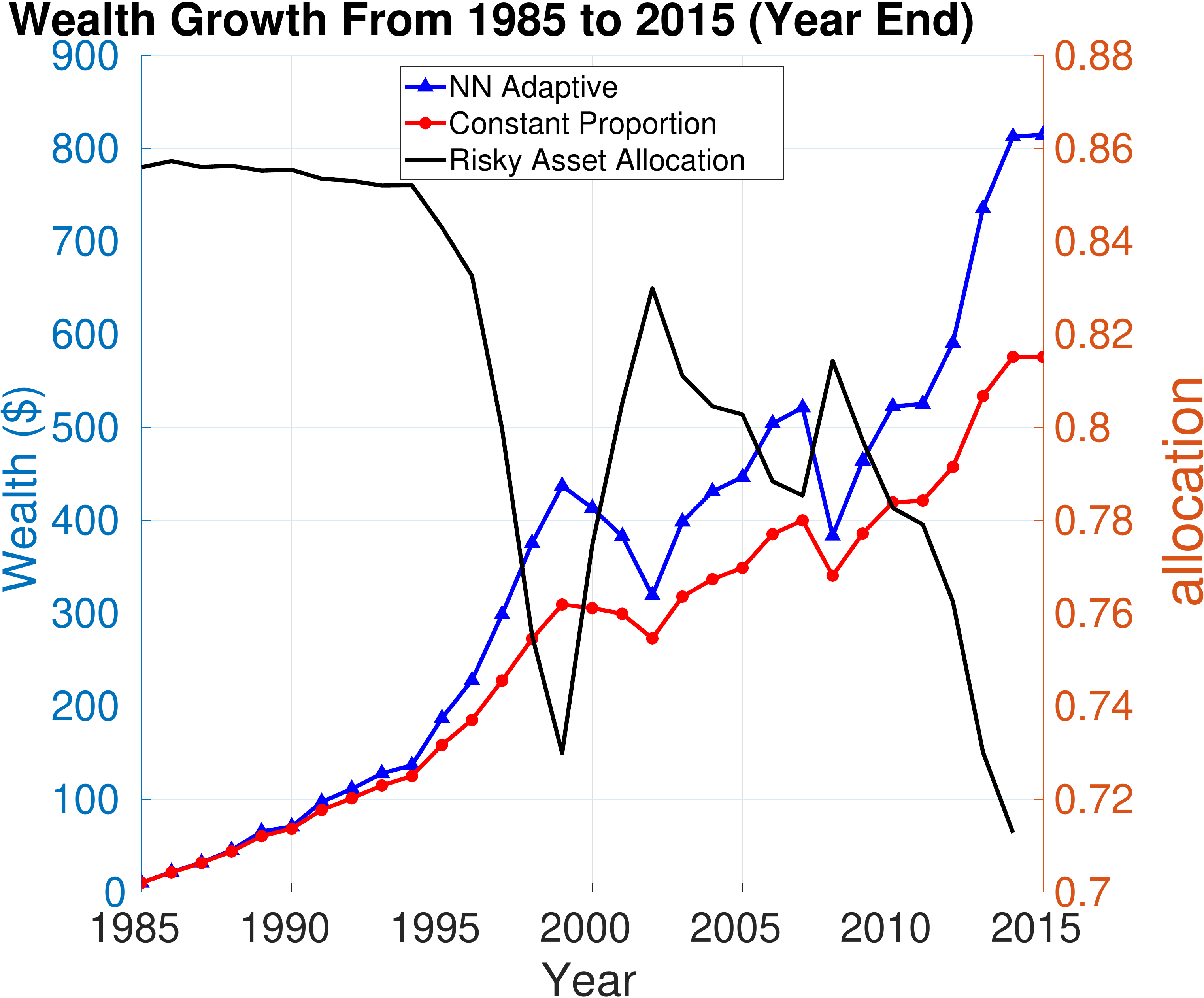}
\caption{Backtest of strategy performance over the historical period from 1985-2015 (single path)}
\label{fig:backtest}
\end{figure}

\subsection{Comparison with the 80/20 Constant Proportion Strategy}
While the average stock allocation from the optimal adaptive strategy varies over time, its average over time is about 80\%.
A natural question is how the optimal adaptive strategy compares with the 80/20 constant proportion strategy which invests 80\% of the wealth in the stocks and 20\% in the bonds.

Here we compare the optimal adaptive strategy with the 80/20 constant proportion strategy.
Recall that in Section \ref{sec:mkt}, the optimal adaptive strategy is trained on bootstrap resampled data with the expected blocksize of $0.5$ years and the test dataset is bootstrap resampled data with the expected blocksize of $2$ years. We compare the optimal adaptive strategy and 80/20 strategy on the same test dataset.

In Figure \ref{fig:8020_cdf_wdiff}, we plot CDFs of $W_{NN}(T)-W_{50/50}(T)$ and $W_{80/20}(T)-W_{50/50}(T)$, i.e., the wealth difference of the optimal adaptive strategy and the 80/20 strategy from the 50/50 strategy respectively.

\begin{figure}[htp]
\centering
\begin{subfigure}[b]{0.47\textwidth}
\centering
\includegraphics[width=\textwidth]{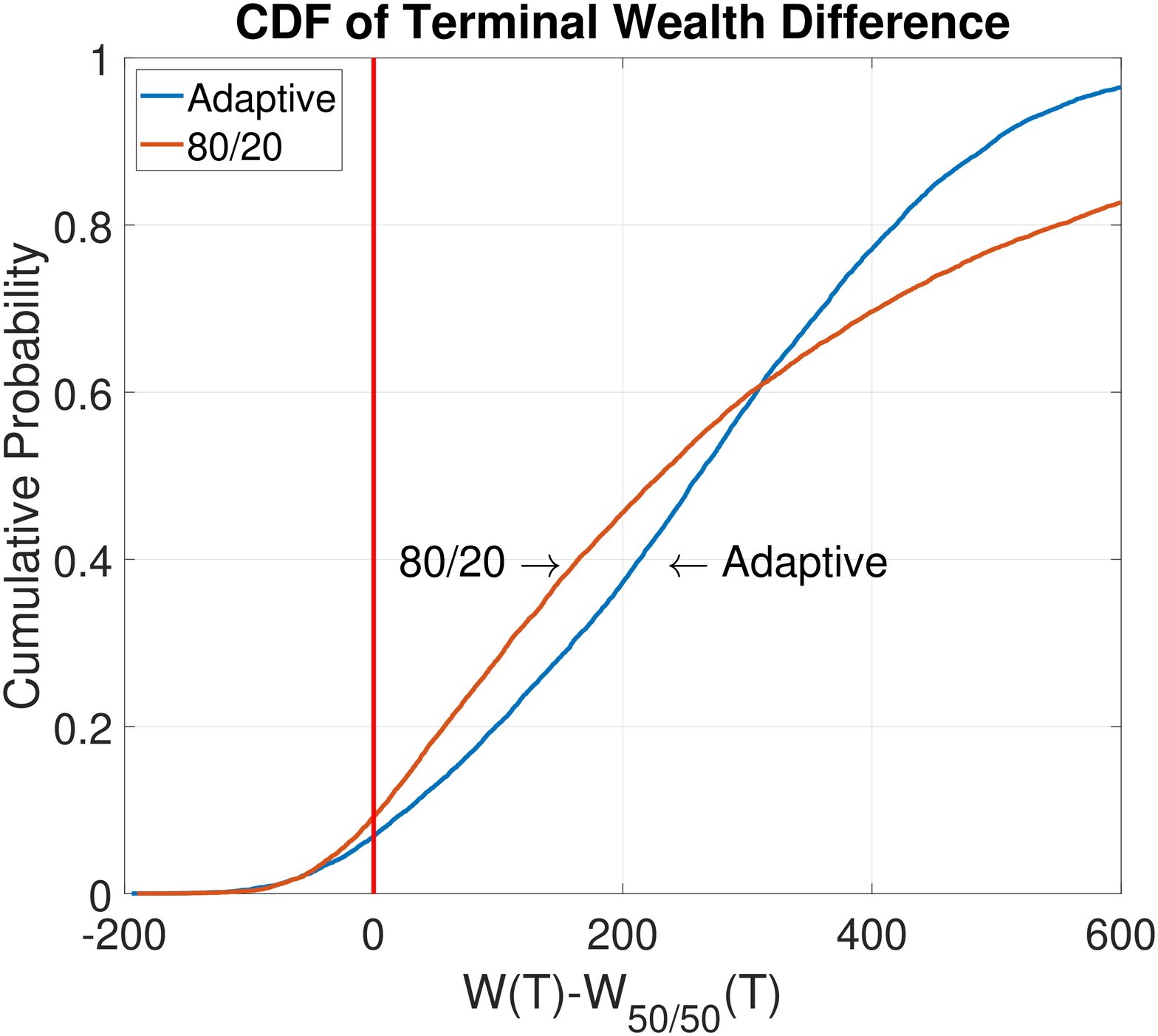}
\caption{CDF of terminal wealth difference $W(T)-W_{50/50}(T)$, $W(T)$ is either $W_{NN}(T)$ or $W_{80/20}(T)$}
\label{fig:wdiff}
\end{subfigure}
\hfill
\begin{subfigure}[b]{0.47\textwidth}
\centering
\includegraphics[width=\textwidth]{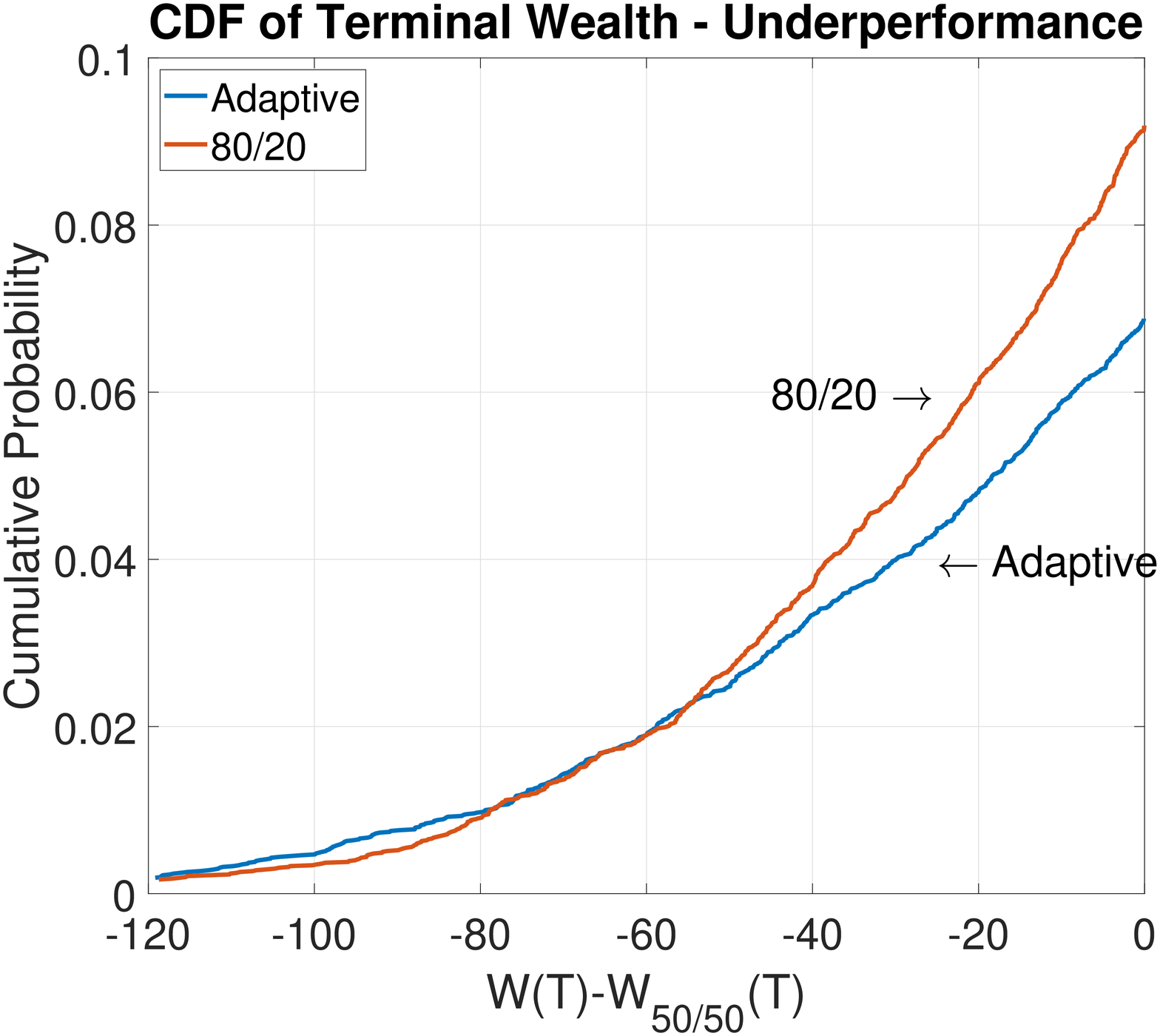}
\caption{CDF of terminal wealth difference - enlarged for underperformance}
\label{fig:wdiff_enlarged}
\end{subfigure}
\caption{CDF of wealth difference of both strategies (optimal adaptive and 80/20 constant proportion) over the 50/50 strategy}
\label{fig:8020_cdf_wdiff}
\end{figure}
We observe that the optimal adaptive strategy controls tail risk better than the 80/20 strategy. Specifically, the probability of the optimal adaptive strategy underperforming the 50/50 strategy is lower than the 80/20 strategy. When underperformance against the 50/50 strategy occurs, the magnitude of underperformance for the optimal adaptive strategy is less than the magnitude of underperformance for the 80/20 strategy, as in Figure \ref{fig:8020_cdf_wdiff}.

It is worth noting that the 80/20 strategy has more upside than the optimal adaptive strategy. However, we should remind the readers that less upside is a natural result of our choice of the double-sided penalty objective function. As reflected in the asymmetric objective function, our goal is not to achieve extremely large outperformance over the 50/50 strategy, but to reach the elevated target with high probability and to control the downside risk. The optimal adaptive strategy achieves those goals better than the 80/20 strategy. To better demonstrate this, we plot the following CDF of outperformance of both strategies over the elevated target $e^{sT}\cdot W_{50/50}(T)$, in Figure \ref{fig:wdiff_cdf_elev_enlarged}.

\begin{figure}[htp]
\centering
\begin{subfigure}[b]{0.48\textwidth}
\centering
\includegraphics[width=\textwidth]{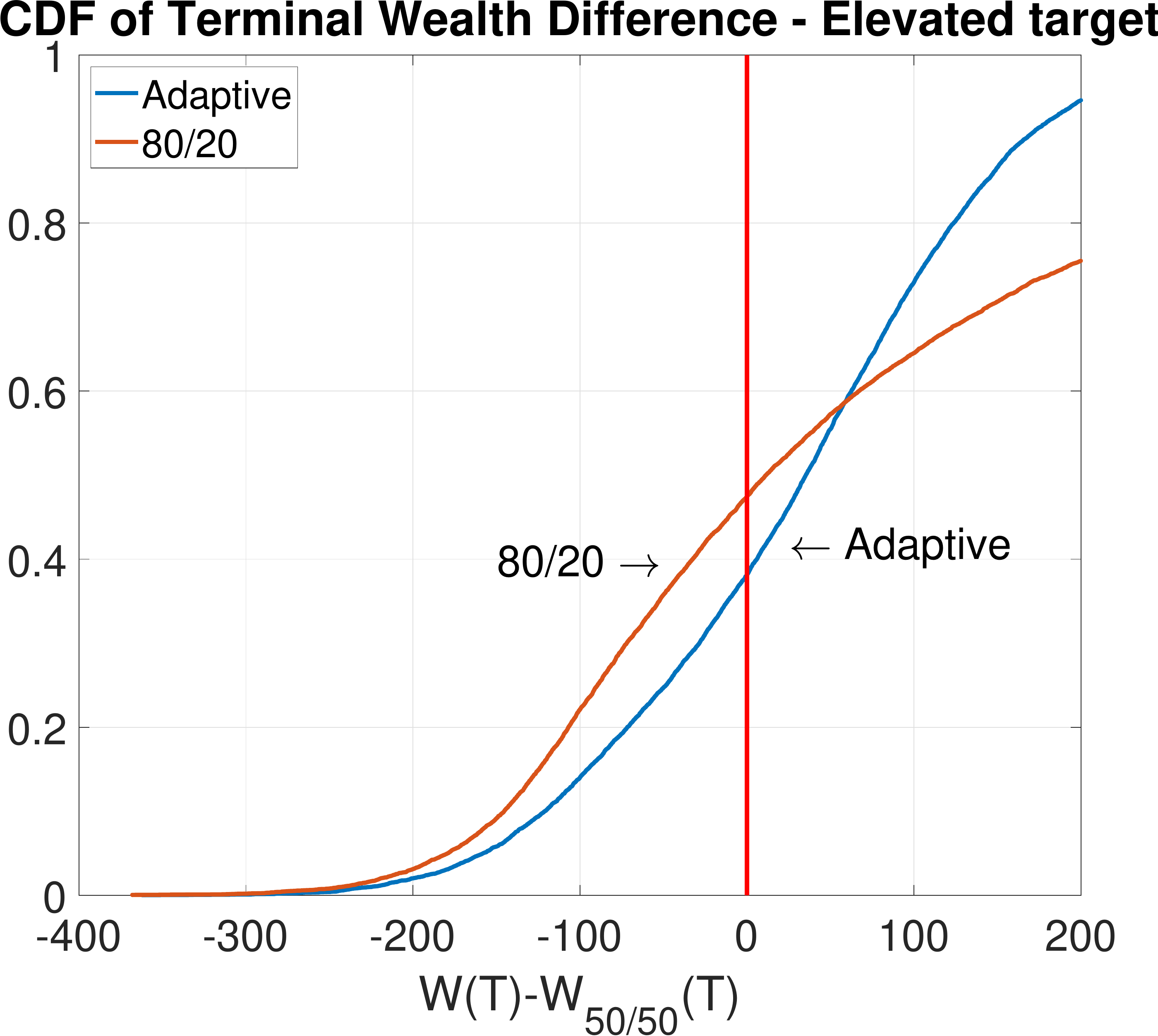}
\caption{CDF of terminal wealth difference $W(T)-e^{sT}\cdot W_{50/50}(T)$, $W(T)$ is either $W_{NN}(T)$ or $W_{80/20}(T)$}
\label{fig:wdiff_cdf_elev}
\end{subfigure}
\hfill
\begin{subfigure}[b]{0.45\textwidth}
\centering
\includegraphics[width=\textwidth]{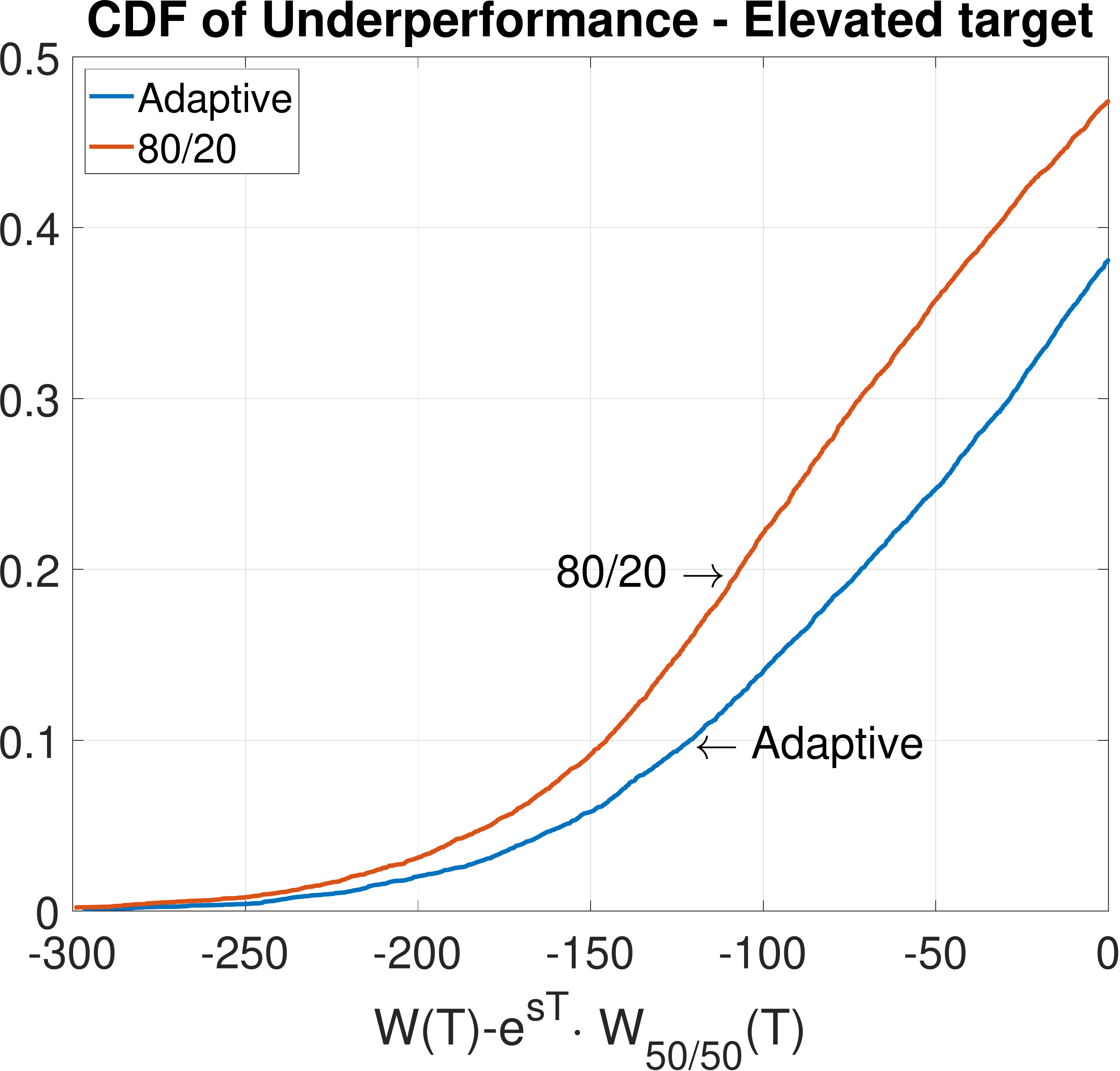}
\caption{CDF of terminal wealth difference over the elevated target - enlarged for underperformance}
\label{fig:wdiff_cdf_elev_enlarged}
\end{subfigure}
\caption{CDF of wealth difference of both strategies (optimal adaptive and 80/20 constant proportion) over the elevated target $e^{sT}\cdot W_{50/50}(T)$}
\label{fig:8020_cdf_wdiff_elevated}
\end{figure}

We also observe that the optimal adaptive strategy has a smaller probability of underperforming the elevated target (37.3\%) than the 80/20 strategy (46.8\%). This means the optimal adaptive strategy is more likely to reach the elevated target and thus achieve the pre-determined annual outperformance spread.

Moreover, we observe from the enlarged CDF plot in Figure \ref{fig:wdiff_cdf_elev_enlarged} that the optimal adaptive strategy consistently controls underperformance better than the 80/20 strategy, in the sense that the optimal adaptive strategy underperforms less than the 80/20 strategy when the elevated target is not met.

\section{Robustness Assessment}

To further evaluate the robustness of the optimal adaptive strategy, we assess optimal control models from the following three perspectives:
\begin{itemize}
\item
We test the strategy learned from the bootstrap data with a given expected blocksize on bootstrap data with multiple different expected blocksizes.

\item
We train the model on a dataset simulated from a synthetic parametric model and test it on the bootstrap resampled dataset.
\item
We train the strategy learned on bootstrap data from one segment of the historical data and test the strategy on bootstrap data from another segment of the historical data.
\end{itemize}

We generate the bootstrap resampled data by sampling directly from the specified historical data sequence for training the optimal control model.

\subsection{Testing Using Different Blocksizes}\label{sec:robustness_blocksize}
We test the adaptive strategy learned on bootstrap resampled data with a given blocksize on bootstrap resampled data with different blocksizes.

For illustration, here we only show the testing results of the strategy learned on bootstrap
resampled data with expected blocksize of $0.5$ years, where test data sets are
bootstrap resampled data with blocksizes ranging from 1-10 years. We note that
training on data sets using a different blocksize, and testing on other blocksizes produces
qualitatively similar results.

\begin{table}[htp]
{\scriptsize
\begin{center}
\resizebox{\columnwidth}{!}{
\begin{tabular}{lccccc} \hline
\multicolumn{6}{c}{ Training Results on Bootstrap Data with Expected Blocksize = 0.5 : Market Cap Weighted} \\ \hline
Strategy & $E(W_T)$ & $std(W_T)$& $median(W_T)$ & $Pr(W_T<median(W_T^{CP}))$ & $Pr(W_T<median(W_T^{NN}))$ \\ \hline
constant proportion($p=0.5$) & 678 & 276 & 624 & 0.50 & 0.86 \\
adaptive & 963 & 474 & 913 & 0.26 & 0.50 \\\hline
\multicolumn{6}{c}{ Testing Results on Bootstrap Data: Market Cap Weighted} \\ \hline
Strategy & $E(W_T)$ & $std(W_T)$& $median(W_T)$ & $Pr(W_T<median(W_T^{CP}))$ & $Pr(W_T<median(W_T^{NN}))$ \\ \hline
\multicolumn{6}{c}{Expected Blocksize $\hat{b}=1$ years} \\ \hline
constant proportion($p=.5$) & 674 & 273 & 624 & 0.50 & 0.84 \\
NN adaptive & 955 & 466 & 909 & 0.27 & 0.50 \\\hline
\multicolumn{6}{c}{Expected Blocksize $\hat{b}=2$ years} \\ \hline
constant proportion($p=.5$) & 676 & 263 & 631 & 0.50 & 0.84 \\
NN adaptive & 958 & 445 & 917 & 0.26 & 0.50 \\\hline
\multicolumn{6}{c}{Expected Blocksize $\hat{b}=5$ years} \\ \hline
constant proportion($p=.5$) & 669 & 244 & 626 & 0.50 & 0.85 \\
NN adaptive & 953 & 409 & 915 & 0.24 & 0.50 \\\hline
\multicolumn{6}{c}{Expected Blocksize $\hat{b}=8$ years} \\ \hline
constant proportion($p=.5$) & 669 & 233 & 632 & 0.50 & 0.87 \\
NN adaptive & 960 & 393 & 928 & 0.23 & 0.50 \\\hline
\multicolumn{6}{c}{Expected Blocksize $\hat{b}=10$ years} \\ \hline
constant proportion($p=.5$) & 667 & 223 & 635 & 0.50 & 0.88 \\
NN adaptive & 961 & 383 & 928 & 0.22 & 0.50 \\\hline
\end{tabular}
}
\end{center}
}
\caption{Terminal wealth statistics of the adaptive strategy trained on bootstrap resampled data with expected blocksize $\hat{b}$ = $0.5$ years. Tested on bootstrap resampled data with blocksizes from $1$ to $10$ years.}
\label{tb:mkt_terminal_robustness}
\end{table}

We can observe from Table \ref{tb:mkt_terminal_robustness} that
\begin{itemize}
\item The mean and the median terminal wealth of the adaptive strategy remain similar across different blocksizes.
\item The adaptive strategy has a more favorable terminal wealth distribution as it is more likely to achieve the terminal wealth higher than the median terminal wealth of the constant proportion strategy.
\end{itemize}

Table \ref{tb:mkt_terminal_robustness} demonstrate that the outperformance of the adaptive strategy over the benchmark strategy is robust across different expected blocksizes. We include more testing results from strategies trained with other expected blocksizes in the Appendix.

\subsection{Strategy Trained on Synthetic Data} \label{sec:syn}
In this section, we generate synthetic data from a parametric model calibrated
to historical data.
We then test the strategy on bootstrap resampled data.
Clearly, the synthetic data from the parametric model will have a different distribution
compared to the resampled data.

\subsubsection{Synthetic Data Generation}
The synthetic data is generated based on a jump-diffusion stochastic process.
Let $S(t)$ and $B(t)$ respectively denote the wealth invested in the stocks and bonds at time $t$, $t\in[0,T]$.
Specifically, we will assume that $S(t)$ represents the unit amount invested in a broad stock market index (CRSP cap-weighted index), while $B(t)$ is the unit amount invested in short term default-free government bonds (in our case, the 3-month T-bill).

Recall that $t^- = t -\epsilon, \epsilon\rightarrow0^+$, i.e. $t^-$ is the instant of time before $t$, and let $\psi$ be a random number representing a jump multiplier. When a jump occurs, $S(t) = \xi S(t^-)$.
Allowing discontinuous jumps lets us explore the effects of severe market crashes on the stock holding,
and nonnormal returns. We assume that $\xi$ follows a double exponential
distribution (\citep{kou2002jump}; \citep{kou2004option}). If a jump occurs, $p_{up}$ is the probability of an upward jump, while $1 - p_{up}$ is the chance of a downward jump. The
density function for $y = \log \xi$ is
\begin{equation} \label{eq:xi}
f(y) = p_{up}\eta_1e^{-{\eta_1}y}\textbf{1}_{y\geq0} + (1 - p_{up})\eta_2e^{{\eta_2}y}\textbf{1}_{y\le0}.
\end{equation}
For future reference, note that
\begin{equation} \label{eq:exp_xi}
E[y=\log\xi]=\frac{p_{up}}{\eta_1}-\frac{(1-p_{up})}{\eta_2},\; E[y=\xi]=\frac{p_{up}\eta_1}{\eta_1-1}+\frac{(1-p_{up})\eta_2}{\eta_2-1}
\end{equation}
We assume that $S(t)$ evolves according to
\begin{equation} \label{eq:S}
\frac{dS(t)}{S(t^-)}=(\mu-\lambda E[\xi-1])dt+\sigma dZ+d\big(\sum_{i=1}^{\pi_t}(\xi_i-1)\big),
\end{equation}
where $\mu$ is the (uncompensated) drift rate, $\sigma$ is the volatility, $dZ$ is the increment of a Wiener process, $\pi_t$
is a Poisson process with positive intensity parameter $\lambda$, and $\xi_i$ are i.i.d. positive random variables having
distribution (\ref{eq:xi}). Moreover, $\xi_i$, $\pi_t$, and $dZ$ are assumed to all be mutually independent.

We assume that the dynamics of the amount $B(t)$ invested in the risk-free asset are
\begin{equation} \label{eq:B}
dB(t) = rB(t)dt,
\end{equation}
where $r$ is the (constant) risk-free rate. This is obviously a simplification of the real bond market. We remind the reader that, ultimately, our NN method is entirely data-driven, and will be based on bootstrapped stock and bond indexes.

Based on (\ref{eq:S}) and (\ref{eq:B}), we use the methods in \citep{dang2016better} to calibrate the process parameters. We use a threshold technique \citep{cont2011nonparametric} to identify jump frequency and distribution, and the methods in \citep{dang2016better} to determine the remaining parameters. Annualized estimated parameters for the cap-weighted stock index is provided in Table \ref{tb:parameters}.

\begin{table}[htp]
\centering
\begin{tabular}{ccccccc}
\hline
$\mu$ & $\sigma$ & $\lambda$ & $p_{up}$ & $\eta_1$ & $\eta_2$ &$r$\\ \hline
\multicolumn{7}{c}{Real CRSP Cap-Weighted Stock Index and 3-month T-bill Index} \\ \hline
.08889 & .14771 & .32222 & 0.27586 & 4.4273 & 5.2613 & 0.00827 \\ \hline
\end{tabular}
\caption{Estimated annualized parameters for double exponential jump diffusion model. Cap-weighted index, deflated by the CPI. Sample period 1926:1 to 2015:12.}
\label{tb:parameters}
\end{table}

We then generate the synthetic data based on the parametric model with the calibrated parameters through Monte Carlo simulations.
\subsubsection{Strategy Performance}
We test the performance of the strategy trained on synthetic data on bootstrap data with expected blocksize $\hat{b}$ = 2 years. Note that the testing performance with other expected blocksizes is very similar to each other so we only show results for $\hat{b}$ = 2 years.

\begin{table}[H]
{\scriptsize
\begin{center}
\resizebox{\columnwidth}{!}{
\begin{tabular}{lccccc} \hline
\multicolumn{6}{c}{ Training Results on Synthetic Data : Market Cap Weighted} \\ \hline
Strategy & $E(W_T)$ & $std(W_T)$& $median(W_T)$ & $Pr(W_T<median(W_T^{CP}))$ & $Pr(W_T<median(W_T^{NN}))$ \\ \hline
constant proportion($p=0.5$) & 714 & 383 & 630 & 0.50 & 0.82 \\
adaptive & 1019 & 651 & 930 & 0.29 & 0.50 \\\hline
\multicolumn{6}{c}{ Testing Results on Bootstrap Data with Expected Blocksize = 2 years} \\ \hline
Strategy & $E(W_T)$ & $std(W_T)$& $median(W_T)$ & $Pr(W_T<median(W_T^{CP}))$ & $Pr(W_T<median(W_T^{NN}))$ \\ \hline
constant proportion($p=0.5$) & 679 & 267 & 630 & 0.50 & 0.84 \\
adaptive & 944 & 431 & 912 & 0.26 & 0.50 \\\hline
\end{tabular}
}
\end{center}
}
\caption{Terminal wealth statistics the adaptive strategy trained on synthetic data and tested on bootstrap resampled data with expected blocksize $\hat{b}$ = 2 years}
\label{tb:syn_terminal}
\end{table}

Table \ref{tb:syn_terminal} shows that the adaptive strategy learned from synthetic data performs well on the test set of bootstrap resampled data. The adaptive strategy have significantly higher median and mean terminal wealth than the constant proportion strategy in both training and testing.

We do notice that in the testing results, the adaptive strategy has
slightly lower mean and median terminal wealth, as well as a lower
standard deviation than in training results. This is hardly surprising since the
training and test data have different distributions. However, overall, the strategy
appears to be quite robust.
Further distribution comparisons can be found in Appendix \ref{append:distr}.

\subsection{Robustness Test With Training/Testing Split}\label{sec:split}
In \S\ref{sec:robustness_blocksize} and \S\ref{sec:syn}, both training and testing datasets are generated from either a parametric model or bootstrap resampled data from a single historical return path from 1926-2015.
A possible criticism of such an approach is that both the training data and testing data share the same information source.
In particular, is it possible for the training data to have a forward-looking bias?

We argue that there is no forward-looking bias in the described training and testing data generation process. Recall that in the experiments, training data and testing data have different expected blocksizes, and thus different distributions. Specifically, when bootstrap resampling randomly with different expected blocksizes, the ordering of blocks of data points is randomly shuffled and
any sequential ordering information is destroyed.
Further, Theorem \ref{thm:fix} and \ref{thm:stb} show that the probability of an entire path in the training dataset
reappearing in the testing dataset is vanishingly small.
This is due to the random block resampling nature of the bootstrap algorithm.

Nonetheless, to provide additional evidence of robustness, we compare the following two different cases:

\begin{description}
\item{{\bf{Case \#1:}}}
We train the adaptive strategy on bootstrap resampled data from the entire historical path from 1926 to 2015. We test the strategy on bootstrap resampled data from the last 30 years of the historical path from 1986-2015. There is an overlap between the underlying historical path for training and testing (1986-2015). We show that such overlap does not introduce an advantage in terms of the strategy performance by comparing it with case \#2 - the {\em non-overlap} case.

\item{{\bf{Case \#2:}}}
We train the adaptive strategy on bootstrap resampled data from the first 60 years of the historical path from 1926 to 1985.
We test the strategy on the same bootstrap resampled data generated from the last 30 years of the historical path from 1986-2015 as in case \#1.
Consequently, there is no overlap between the underlying historical paths we use for generating training data and testing data at all.
\end{description}

Figure \ref{fig:split_case1} and Figure \ref{fig:split_case3} show these two cases schematically. Case \#2 is the more stringent test case as there are zero overlaps between the underlying historical data for the generation of the training set and the testing set.

\begin{figure}[htp]
\centering
\begin{tikzpicture}[y=1cm, x=1cm, thick, font=\footnotesize]
\draw[line width=1.2pt, ->, ](2,0) -- (9.5,0) node[anchor=west] {Year};
\draw (2.5,0.5em) -- (2.5,-0.5em) node[below] {$1926$};
\draw (6.5,0.5em) -- (6.5,-0.5em) node[below] {$1985$};
\draw (8.5,0.5em) -- (8.5,-0.5em) node[below] {$2015$};
\draw [decorate,decoration={brace,amplitude=10pt},xshift=-4pt,yshift=0pt]
(2.63,0.7em) -- (8.62,0.7em) node [black,midway,yshift=0.6cm]
{\footnotesize Training};

\draw [decorate,decoration={brace,amplitude=10pt, mirror},xshift=-4pt,yshift=0pt]
(6.63,-1.5em) -- (8.62,-1.5em) node [black,midway,yshift=-0.6cm]
{\footnotesize Testing};
\end{tikzpicture}
\caption{Case \#1: use historical data from 1926-2015 for generating training data, and 1986-2015 for testing. There is an overlap between the underlying historical paths for training and testing.}
\label{fig:split_case1}
\end{figure}
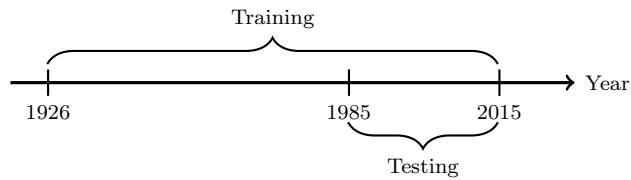

\begin{figure}[htp]
\centering
\begin{tikzpicture}[y=1cm, x=1cm, thick, font=\footnotesize]
\draw[line width=1.2pt, ->, ](2,0) -- (9.5,0) node[anchor=west] {Year};
\draw (2.5,0.5em) -- (2.5,-0.5em) node[below] {$1926$};
\draw (6.5,0.5em) -- (6.5,-0.5em) node[below] {$1985$};
\draw (8.5,0.5em) -- (8.5,-0.5em) node[below] {$2015$};
\draw [decorate,decoration={brace,amplitude=10pt},xshift=-4pt,yshift=0pt]
(2.63,0.7em) -- (6.62,0.7em) node [black,midway,yshift=0.6cm]
{\footnotesize Training};

\draw [decorate,decoration={brace,amplitude=10pt, mirror},xshift=-4pt,yshift=0pt]
(6.63,-1.5em) -- (8.62,-1.5em) node [black,midway,yshift=-0.6cm]
{\footnotesize Testing};
\end{tikzpicture}
\caption{Case \#2: ``non-overlap'' case where underlying market data for training and testing data has no overlaps. Case \#2 uses the same testing dataset as case \#1.}
\label{fig:split_case3}
\end{figure}
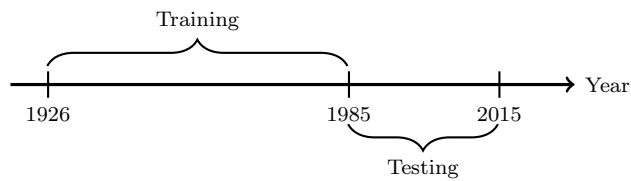

Note that, for Case \#1 and \#2, the underlying historical data for testing data has only a 30-year window. Recall that our investment horizon
in our previous experiments was $T=30$ years. In order to obtain more meaningful block bootstrap resampling results for
the non-overlap window, we will reduce the investment horizon to $T=15$ years, for both cases in this section.

We first compare the CDF of the terminal wealth of the two cases.
\begin{figure}[htp]
\centering
\includegraphics[width=0.6\linewidth]{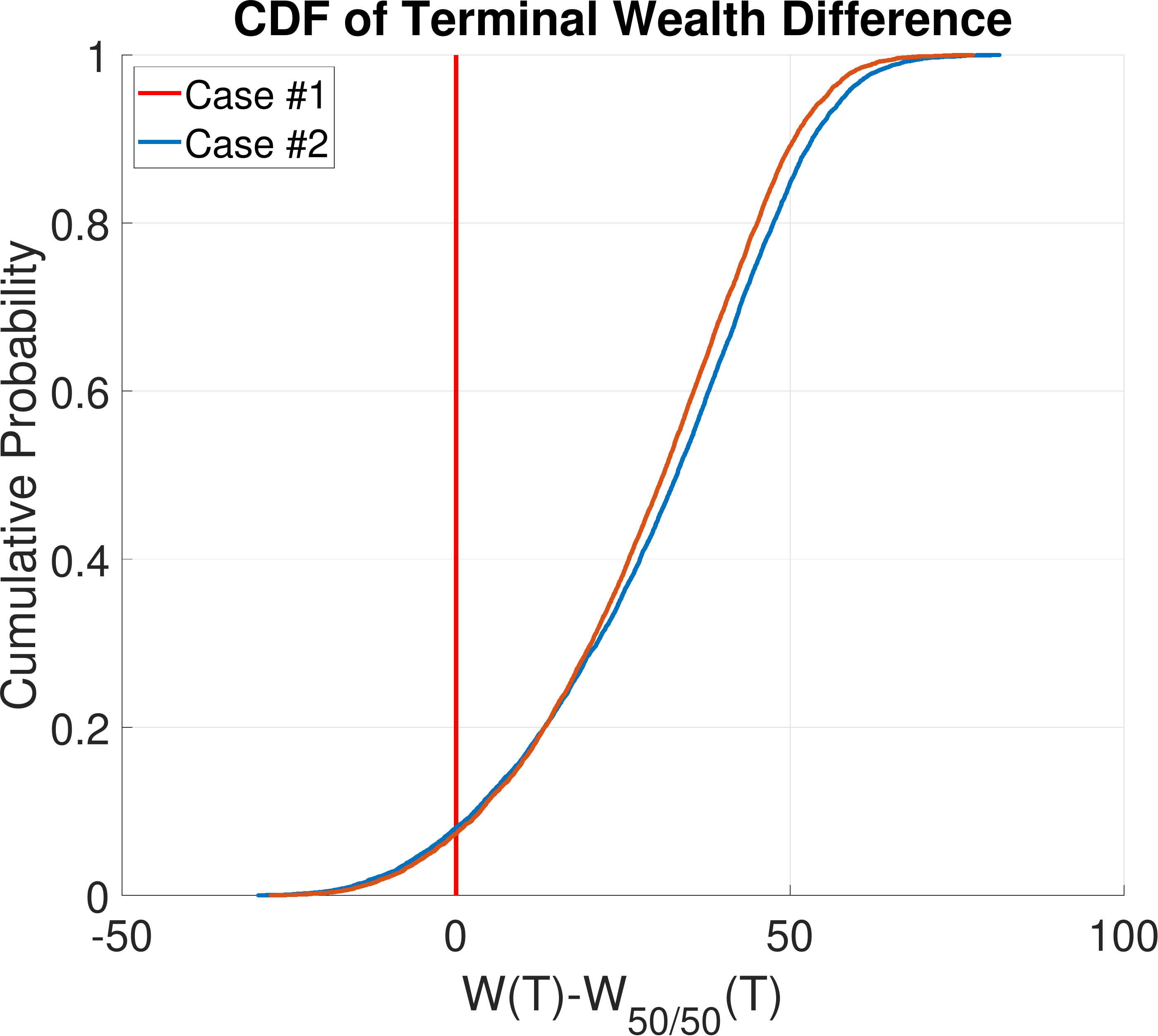}
\captionof{figure}{CDF of wealth difference $W(T)-W_{50/50}(T)$ for the two cases: case \#1: train: 1926-2015, test: 1986-2015; case \#2: train: 1926-1985, test: 1986-2015.}
\label{fig:split_wdiff_cdf}
\end{figure}%
From Figure \ref{fig:split_wdiff_cdf} we can observe that Case \#1 and Case \#2 have almost identical CDF curves. The almost identical CDF curves for Case \#1 and Case \#2 (the {\em non-overlap} case) - supports our argument that forward-looking bias is not a concern in our approach. Despite using the entire historical period as the underlying data for training, case \#1 does not have a superior CDF than Case \#2, in which the underlying market data for training data and testing data have no overlaps.

\begin{figure}[htp]
\centering
\includegraphics[width=1\linewidth]{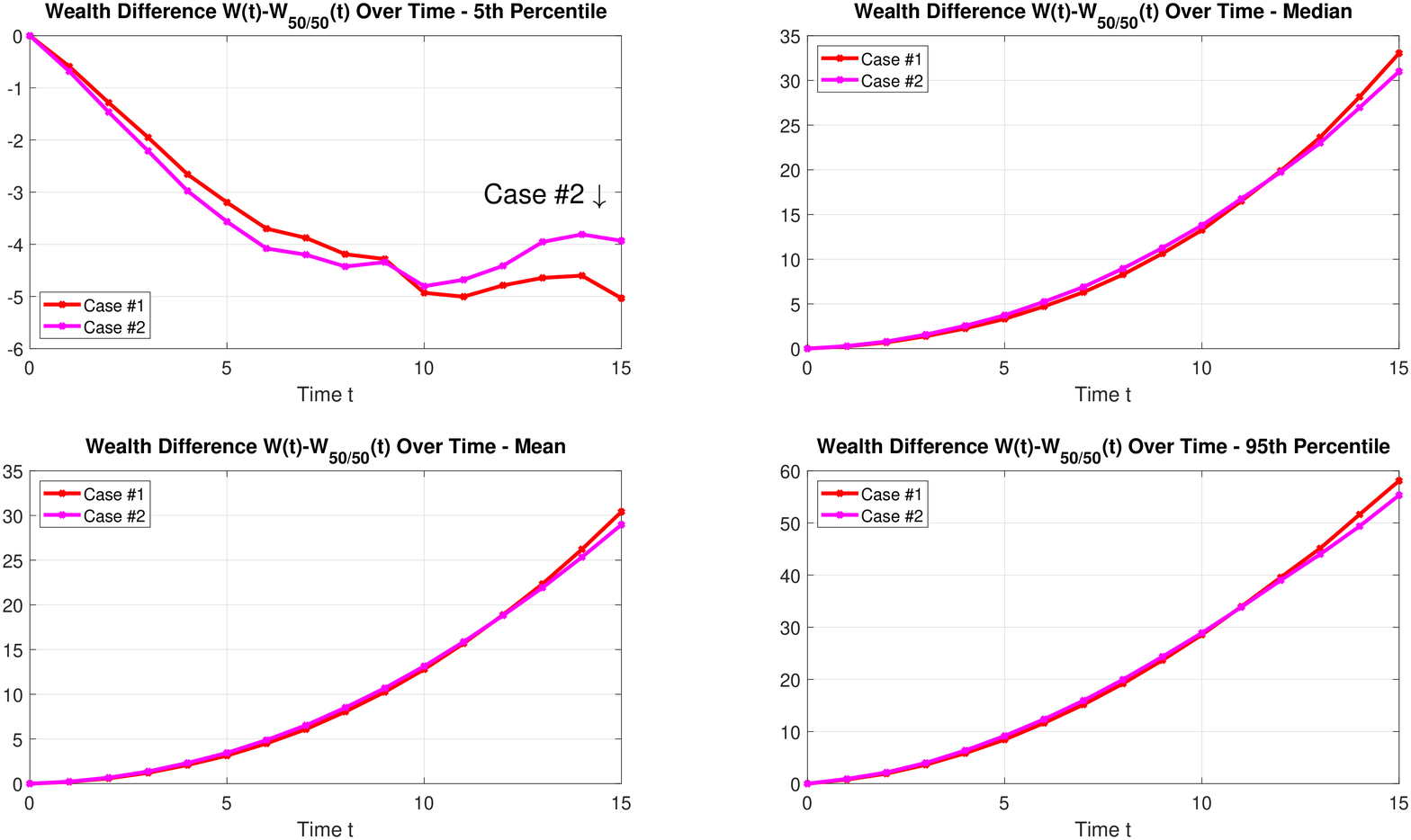}
\captionof{figure}{Percentiles of wealth difference $W(T)-W_{50/50}(T)$ for the two cases}
\label{fig:split_wdiff}
\end{figure}%

In Figure \ref{fig:split_wdiff}, we show the different percentiles of wealth difference between the adaptive strategy and the constant proportion strategy for both cases.

Again, Case \#1 and Case \#2 have almost identical performances, except that Case \#2 has slightly better tail risk control than Case \#1 (5th percentile). This further proves that the overlap does not introduce performance advantage as the {\em non-overlap} case actually has less tail risk.

\begin{figure}[htp]
\centering
\includegraphics[width=1\linewidth]{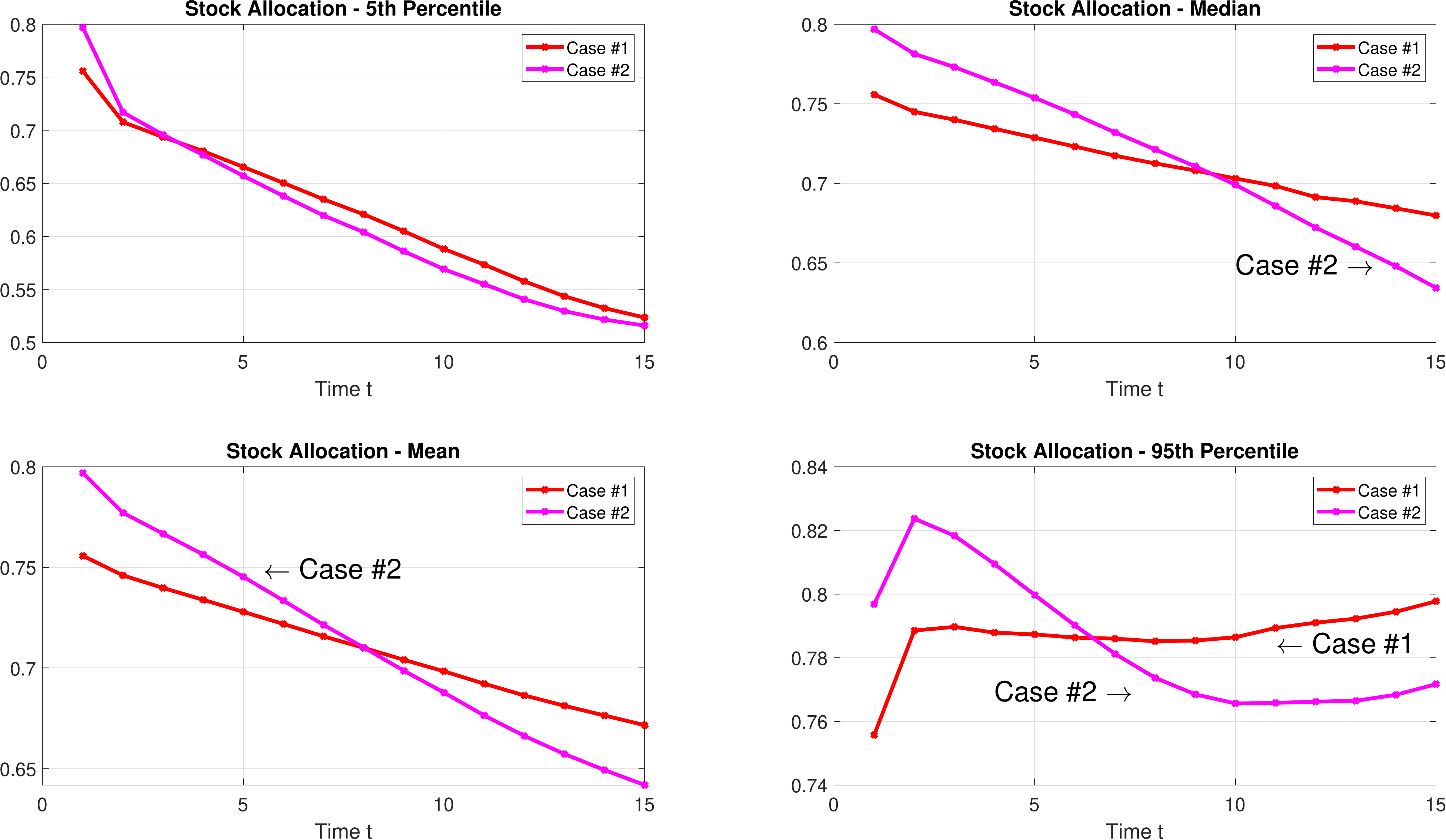}
\captionof{figure}{Stock allocation for the two cases}
\label{fig:split_alloc}
\end{figure}

In Figure \ref{fig:split_alloc}, we compare the actual strategies, i.e., stock allocations of both cases. This time we can observe some differences between Case \#1 and Case \#2. From the median and mean plot, we can observe that Case \#2 tends to derisk (decrease allocation in the stocks) more aggressively over time than Case \#1. We believe the difference comes from the difference in the distributions between the different segments of the underlying historical market returns. However, the difference between allocation strategies is not significant. In fact, the average stock holding over time are quite similar for both cases. In addition, we have already observed similar strategy performances in terms of terminal wealth distributions from figure \ref{fig:split_wdiff_cdf} and figure \ref{fig:split_wdiff}.

In conclusion, the results further illustrate the robustness of our approach and show that forward-looking bias is not a concern in our work.

\section{Conclusions}
In this article, we propose a data-driven framework for computing the optimal asset allocation for outperforming a stochastic benchmark target based on market asset return observations.
The scenario-based dynamic asset allocation problem is solved directly
assuming a neural network representation for the optimal control, without
using dynamic programming. This leads to a method that avoids the curse of
dimensionality which is a critical issue in dynamic allocation for outperforming a stochastic benchmark.

In addition, we design an asymmetric distribution shaping objective function which
is capable of producing an optimal strategy which can yield significantly
larger median terminal wealth than the target, with only a small probability (and magnitude) of underperformance.
We emphasize that our methodology can encompass a wide class of objective functions, which can
be tailored to the risk preferences of individual investors.

We use block bootstrap resampling to augment historical financial market data.
The training data is generated by block bootstrap resampling from market asset returns. This leads
to a data-driven approach for determining the optimal dynamic asset allocation,
avoiding the need to make a parametric asset price model as well as model parameter estimations.
We further provide mathematical justifications for using block bootstrap resampling to generate both training and testing datasets.

The proposed method is illustrated in the DC pension allocation problem,
which is a practically relevant and important problem on its own.
We evaluate and analyze the performance of the optimal NN adaptive
strategy based on CRSP
3-month Treasury bill (T-bill) index for the risk-free asset
and the CRSP cap-weighted total return index for the risky asset from {\em 1926:1-2015:12}.

We illustrate the robustness of our approach from three different perspectives.
\begin{itemize}
\item We show that the adaptive strategy trained on bootstrap resampled data with a given expected
blocksize performs consistently well on bootstrap resampled data with different expected blocksizes (thus different distributions).

\item We show that the adaptive strategy learned on synthetic data performs well on bootstrap resampled data,
despite the fact that the methodology for generating the datasets are quite different.

\item We compare the performance of our strategy with the strategy
learned in an {\em non-overlap} setting where the underlying market data for the
training dataset and testing dataset has no overlap. We show that the {\em non-overlap}
case has a comparable performance which supports our argument that
forward-looking bias should not be a concern in our approach.
\end{itemize}

Basing our optimal control on a shallow Neural Network representation using only a small number of
financially relevant feature variables results in a strategy that is financially intuitive and implementable.

\section{Acknowledgements}
This work was supported by a Collaborative Research and
Development (CRD) grant from the Natural Sciences and Research Council of Canada
NSERC Neuberger Berman CRD: \#50492-10196 -2950-105
and by a grant from Neuberger Berman.

\section{Conflicts of interest}
The authors have no conflicts of interest to report.

\appendix
\section{Appendix}
\subsection{Proofs for Theorem \ref{thm:fix} and \ref{thm:stb} } \label{appendix:proof}

We mathematically establish Theorem \ref{thm:fix} and \ref{thm:stb}.

For a path $\ppath$, we use the following notations:
\begin{eqnarray}
\hat{b} & = & {\mbox{expected blocksize in stationary block bootstrap}} \nonumber \\
N & = & {\mbox{number of total datapoints in the path}} \nonumber \\
N_{tot} & = & {\mbox{number of total datapoints to bootstrap from}} \nonumber \\
\ppath[i] & = & {\mbox{the $i$th data point in path $\ppath$}} \nonumber \\
\end{eqnarray}

We also make the following definitions.
\begin{definition}
Assume that a path $\ppath$ of length $N$, which contains blocks $[B_1,\ldots B_k] $,
is resampled from the original data path of length $N_{tot}$.
The \textbf{decision index list} $[I_1,\ldots, I_k] $ of the path $\ppath$ is defined
as the list of starting indices of every block in the resampled path with
$ I_1=1$, $I_i = 1 + \sum_{j=1}^{i-1} |B_j|$, $i=2,\ldots k$, where $|B_j|$ denotes the number of points in the block $B_j$.
If $I_k$ is the starting index of the last block in the path, then, for index completeness, we define $I_{k+1} \equiv N+1$.
\end{definition}

\begin{remark}[Decision Index List Example]
Given a decision index list $[I_1,\ldots, I_k] $, associated with a path $\ppath$, then the data point of the path, which starts
at decision index $I_i$, is $\ppath[ I_i ]$.
\end{remark}

\begin{definition}
For any two paths $\ppath_1$ and $\ppath_2$, the \textbf{combined decision index list}
of $\ppath_1$ and $\ppath_2$ is the merged index list (with only a single copy of each index) of
the decision index lists of $\ppath_1$ and $\ppath_2$. The merged list $[I_1, \ldots, I_p]$ retains
the order properties of the original lists, i.e. $I_{i+1} > I_i$ and $I_{p+1} =N+1$.
\end{definition}

\begin{definition}
For any two paths $\ppath_1$ and $\ppath_2$, we define $N_{cdi}(\ppath_1,\ppath_2)$ as
the length of the combined decision index list of $\ppath_1$ and $\ppath_2$.
\end{definition}

\begin{lemma}\label{Lemma:LemmaA1}
Consider either the fixed block resampling or stationary resampling from a sequence of $N_{tot}$ distinct observations.
Two paths $\ppath_1$ and $\ppath_2$ with $[I_1,I_2,\ldots,I_{cdi}]$ as the combined decision index list are identical if and only $\ppath_1[I_j]=\ppath_2[I_j]$ at any $I_j$, $j=1,\ldots, N_{cdi}$.
\end{lemma}

\begin{proof}
First, $\ppath_1$ equals to $\ppath_2$ clearly implies that $\ppath_1[I_j]=\ppath_2[I_j]$ at any $I_j$, $j=1,\ldots, N_{cdi}$.
Conversely, assume that $\ppath_1[I_j]=\ppath_2[I_j], j=1,\ldots,N_{cdi}$.
For any $j, j=1,\ldots,N_{cdi}$, the entire segment $ \ppath_1[I_j],\ldots,\ppath_1[I_{j+1}-1]$ is from the
same resampled subblock of the original data.
Similarly, the the entire segment $ \ppath_2[I_j],\ldots,\ppath_2[I_{j+1}-1]$ is from the
same resampled subblock of the original data.
Since $ \ppath_1[I_j]= \ppath_2[I_j]$, then $ \ppath_1[I_j],\ldots,\ppath_1[I_{j+1}-1]$ and
$ \ppath_2[I_j],\ldots,\ppath_2[I_{j+1}-1]$ are identical.
Thus, the entire paths $\ppath_1$ and $\ppath_2$ are identical.

\end{proof}

~\\
\textbf{THEOREM \ref{thm:fix}.}
Consider fixed block resampling sequences of $N$ points from a sequence of $N_{tot}$ distinct observations .
Let path $\ppath_1$ be a bootstrap resampled path with a fixed blocksize of $b_1$ and path $\ppath_2$ be a
bootstrap resampled path with a fixed blocksize of $b_2$. Then the probability of $\ppath_1$ and $\ppath_2$
being identical is $(\frac{1}{N_{tot}})^{lcm(\frac{N}{b_1},\frac{N}{b_2})}$,
where $lcm(a,b)$ is the least common multiple of integer $a,b$.

~\\
\noindent
\begin{proof}

Let $I$ denote the combined decision index list of $\ppath_1$ and $\ppath_2$, with
$N_{cdi}$ the total number of combined decision points and $I_j$ denoting the $j$th index within $I$.

From Lemma \ref{Lemma:LemmaA1}, two paths are identical if and only if
$\ppath_1[I_j]=\ppath_2[I_j]$ at any $I_j$, $j=1,\ldots, N_{cdi}$.

For any $j=1,\ldots, N_{cdi}$, since each starting point of either $\ppath_1$ or $\ppath_2$ is chosen independently with equal probability
$\mathbb{P}(\ppath_1[I_j]=\ppath_2[I_j])= \frac{1}{N_{tot}}$.
In addition
\begin{align}
\mathbb{P}(\ppath_1[I_j]=\ppath_2[I_j],j=1,\ldots,N_{cdi}(\ppath_1,\ppath_2))&=\displaystyle\prod_{j=1}^{N_{cdi}(\ppath_1,\ppath_2)}\mathbb{P}(\ppath_1[I_j]=\ppath_2[I_j]) \nonumber \\
&=(\frac{1}{N_{tot}})^{N_{cdi}(\ppath_1,\ppath_2)} \nonumber
\end{align}.

Since $N_{cdi}(\ppath_1,\ppath_2)=lcm(\frac{N}{b_1},\frac{N}{b_2})$, the probability of $\ppath_1$ and $\ppath_2$ being identical is $(\frac{1}{N_{tot}})^{lcm(\frac{N}{b_1},\frac{N}{b_2})}$.
\end{proof}

Next, we consider the stationary block bootstrap resampling, in which the blocksizes are randomly generated from a
shifted geometric distribution.

\begin{property}[Properties of a Geometric Distribution]\label{Lemma:individual_block_prob}
Suppose the integer $m>0$ is drawn from a shifted geometric distribution, with $\mathbb{E}[m] = 1/p$, then
\begin{eqnarray}
\mathbb{P}[ m = k] & = & (1-p)^{k-1} p \nonumber \\
\mathbb{P} [ m \geq k] & = & (1-p)^{k-1} ~. \label{prop_geom}
\end{eqnarray}
We rewrite equation (\ref{prop_geom}) in a form amenable to manipulation. Let
\begin{eqnarray}
(1-p) & = & e^{- \lambda}~,
\end{eqnarray}
so that equation (\ref{prop_geom}) becomes
\begin{eqnarray}
\mathbb{P} [ m=k] &= &e^{-\lambda k}(e^{\lambda} -1 ) \nonumber \\
\mathbb{P} [ m \geq k] & = & e^{ -\lambda (k-1)} \nonumber \\
\lambda &= &- \log[ 1 - p] ~. \label{prop_geom_2}
\end{eqnarray}
Denote the expected blocksize by $\hat{b}$, then in our case, $p = 1/\hat{b}$,
and consequently
\begin{eqnarray}
\lambda = - \log\biggl[ 1 - \frac{1}{\hat{b}} \biggr] ~. \label{lambda_to_hatb}
\end{eqnarray}
\end{property}

\begin{lemma}\label{Lemma:di_prob}
Suppose $[I_1,\ldots,I_k]$ be the decision index list of a block resampled path of length $N$ with the expected blocksize of $\hat{b}$.
Then the probability of the decision index list $[I_1,\ldots,I_k]$ occurring is
$e^{-\lambda (N-1)}(e^{\lambda}-1)^{k-1}$, with $\lambda = - \log[ 1 - \frac{1}{\hat{b}} ]$.
\end{lemma}
\begin{proof}
By definition, $I_{j+1}>I_j$ for any $j=1,\ldots, k-1$,
and $I_{k+1}=N+1$.
The probability of path $\ppath$ having $[I_1,\ldots,I_k]$ as the decision index list
is equal to the probability of path $\ppath$ having the first block with blocksize of $I_2-I_1$, $\ldots$, the $k$th block with blocksize of $I_{k+1}-I_{k}$. Denote the blocks of path $\ppath$ as $B_1,\ldots,B_k$. According to Properties \ref{Lemma:individual_block_prob},
\[
\mathbb{P}(blocksize(B_j)=I_{j+1}-I_j) =
\begin{cases}
e^{-\lambda (I_{j+1}-I_j)}(e^{\lambda}-1), & \text{if } j <k \\
e^{-\lambda (I_{k+1}-I_k-1)}, & \text{if } j=k
\end{cases}
\]
The probability of path $\ppath$ having $[I_1,\ldots,I_k]$ as the decision index
list is
$$\displaystyle\prod_{j=1}^k \mathbb{P}( blocksize(B_j) = I_{j+1}-I_j )
=
e^{-\lambda (I_{k+1}-I_1-1)}
(e^{\lambda }-1)^{k-1}
=
e^{- \lambda (N-1)}(e^{\lambda }-1)^{k-1}.$$
\end{proof}

Lemma \ref{Lemma:di_prob} shows that the probability of a stationary block resampled path $\ppath$ with an expected blocksize of $\hat{b}$ having a decision index list is uniquely determined by the expected blocksize $\hat{b}$, the path length $N$, and the length $k$ of the decision index list.

\begin{lemma}\label{Lemma:cdi_prob}
Suppose two paths $\ppath_1$ and $\ppath_2$ of the length $N$ are generated by stationary block
bootstrap resampling with the expected blocksizes of $\hat{b}_1$ and $\hat{b}_2$ respectively.
Then
\begin{eqnarray}
\mathbb{P}(N_{cdi}(\ppath_1,\ppath_2)=k)
& = &\binom{N-1}{k-1}
e^{-( \lambda_1 + \lambda_2 )(N-1)}
( e^{ \lambda_1 + \lambda_2 } - 1 )^{k-1}
\nonumber \\
& & \lambda_1 = - \log\biggl[ 1 - \frac{1}{\hat{b}_1} \biggr] ~;~
\lambda_2 = - \log \biggl[ 1 - \frac{1}{\hat{b}_2} \biggr] ~.
\end{eqnarray}
\end{lemma}

\begin{proof}

Let $f(\hat{b},n)$ denote the occurrence probability of a stationary block resampled path of
length $N$ with the expected blocksize of $\hat{b}$ and a decision index list of length $n$ (this is given
by Lemma \ref{Lemma:di_prob}).

Suppose $[I_1,\ldots,I_k]$ is a combined index list of any two paths $\ppath_1$ and $\ppath_2$.
Let $v$ be the number of overlapped indices and $i$ be the number of non-overlapped indices for $\ppath_1$ respectively, corresponding to $[I_1,\ldots,I_k]$.

Enumerating the possible values for $v$, the number of overlapped indices and values for $i$, the number non-overlapped indices in $\ppath_1$,
the probability of a combined decision index list $[I_1,\ldots,I_k]$ occurring equals
\begin{equation}
\sum_{v=1}^k\Big(\binom{k-1}{v-1}\sum_{i=0}^{k-v}\binom{k-v}{i}f(\hat{b}_1,v+i)f(\hat{b}_2,k-i)\Big).\label{Eq:Eq_cdi}
\end{equation}

Note that
\begin{align*}
&\sum_{v=1}^k\Big(\binom{k-1}{v-1}\sum_{i=0}^{k-v}\binom{k-v}{i}f(\hat{b}_1,v+i)f(\hat{b}_2,k-i)\Big)\\
=&\sum_{v=1}^k\Big(\binom{k-1}{v-1}\sum_{i=0}^{k-v}\binom{k-v}{i}e^{-\lambda_1(N-1)}(e^{\lambda_1}-1)^{v+i-1}e^{-\lambda_2(N-1)}(e^{\lambda_2}-1)^{k-i-1}\Big)\\
=&e^{-(\lambda_1+\lambda_2)(N-1)}\sum_{v=1}^k\Big(\binom{k-1}{v-1}\big(e^{\lambda_1+\lambda_2}-e^{\lambda_1}-e^{\lambda_2}+1\big)^{v-1}\Big(\sum_{i=0}^{k-v}\binom{k-v}{i}(e^{\lambda_1}-1)^{i}(e^{\lambda_2}-1)^{k-v-i}\Big)\Big)\\
=&e^{-(\lambda_1+\lambda_2)(N-1)}\sum_{v=1}^k\Big(\binom{k-1}{v-1}\big(e^{\lambda_1+\lambda_2}-e^{\lambda_1}-e^{\lambda_2}+1\big)^{v-1}\Big(e^{\lambda_1}+e^{\lambda_2}-2\Big)^{k-v}\Big)\\
=&e^{-(\lambda_1+\lambda_2)(N-1)}(e^{\lambda_1+\lambda_2}-1)^{k-1}\\
\end{align*}

Since there are $\binom{N-1}{k-1}$ combinations of the decision index list of length $k$, we conclude
$$\mathbb{P}(N_{cdi}(\ppath_1,\ppath_2)=k)=\binom{N-1}{k-1}e^{-(\lambda_1+\lambda_2)(N-1)}(e^{\lambda_1+\lambda_2}-1)^{k-1}.$$
\end{proof}

Using Lemma \ref{Lemma:LemmaA1} and Lemma \ref{Lemma:cdi_prob}, we establish the probability of two paths generated with stationary block bootstrap resampling being identical.

~\\
\noindent
\textbf{THEOREM \ref{thm:stb}.} \label{thm:stba}
Let $\ppath_1$ and $\ppath_2$ be two paths of the length $N$ generated from the stationary block bootstrap resampling from a sequence of $N_{tot}$ distinct observations with the expected blocksizes of $\hat{b}_1$ and $\hat{b}_2$ respectively. The probability of $\ppath_1$ and $\ppath_2$ being identical is

$$\frac{1}{N_{tot}}\Big(\big(1-\frac{1}{\hat{b}_1}\big)\big(1-\frac{1}{\hat{b}_2}\big)+\frac{\frac{1}{\hat{b}_1}+\frac{1}{\hat{b}_1}-\frac{1}{\hat{b}_1\hat{b}_2}}{N_{tot}}\Big)^{N-1}.$$

~\\

\begin{proof}
Using Lemma \ref{Lemma:LemmaA1}, $\ppath_1=\ppath_2$ if and only if the observations from $\ppath_1$ and $\ppath_2$ are equal at each of the index in the combined decision index list.
Thus
$$
\mathbb{P}\big(\ppath_1=\ppath_2|N_{cdi}(\ppath_1,\ppath_2)=k\big)=\left(\frac{1}
{N_{tot}}\right)^{k}.
$$

Additionally, following Lemma \ref{Lemma:cdi_prob}, we have
\begin{align*}
\mathbb{P}(\ppath_1=\ppath_2)
=&\sum_{k=1}^{N}\mathbb{P}\big(N_{cdi}(\ppath_1,\ppath_2)=k\big)\cdot\mathbb{P}\big(\ppath_1=\ppath_2|N_{cdi}(\ppath_1,\ppath_2)=k\big)\\
=&\sum_{k=1}^{N}\binom{N-1}{k-1}e^{-(\lambda_1+\lambda_2)(N-1)}(e^{\lambda_1+\lambda_2}-1)^{k-1}(\frac{1}{N_{tot}})^k\\
=&\frac{e^{-(\lambda_1+\lambda_2)(N-1)}}{N_{tot}}\sum_{k=1}^{N}\binom{N-1}{k-1}\Big(\frac{e^{\lambda_1+\lambda_2}-1}{N_{tot}}\Big)^{k-1}\\
=&\frac{e^{-(\lambda_1+\lambda_2)(N-1)}}{N_{tot}}\Big(1+\frac{e^{\lambda_1+\lambda_2}-1}{N_{tot}}\Big)^{N-1}\\
=&\frac{1}{N_{tot}}\Big(e^{-(\lambda_1+\lambda_2)}+\frac{1-e^{-(\lambda_1+\lambda_2)}}{N_{tot}}\Big)^{N-1}\\
=&\frac{1}{N_{tot}}\Big(\big(1-\frac{1}{\hat{b}_1}\big)\big(1-\frac{1}{\hat{b}_2}\big)+\frac{\frac{1}{\hat{b}_1}+\frac{1}{\hat{b}_1}-\frac{1}{\hat{b}_1\hat{b}_2}}{N_{tot}}\Big)^{N-1}.
\end{align*}
\end{proof}

\subsection{Additional Robustness Testing Results}
As mentioned in section \ref{sec:mkt}, we only showed terminal wealth statistics for the strategy trained with bootstrap resampled with expected blocksize $\hat{b}=0.5$ years. Here we show the testing performance of strategies trained on bootstrap data with different blocksizes on different testing sets (bootstrap resampled from different blocksizes). The results show that the adaptive strategy consistently outperforms the constant proportion strategy.

\begin{table}[H]
{\scriptsize
\begin{center}
\begin{tabular}{lccccc} \hline
\multicolumn{6}{c}{ Test Results: Market Cap Weighted} \\ \hline
Strategy & $E(W_T)$ & $std(W_T)$& $median(W_T)$ & $Pr(W_T<median(W_T^{CP}))$ & $Pr(W_T<median(W_T^{NN}))$ \\ \hline
\multicolumn{6}{c}{Expected Blocksize $\hat{b}=0.5$ years} \\ \hline
constant proportion($p=.5$) & 678 & 286 & 623.07 & 0.50 & 0.81 \\
NN adaptive & 949 & 478 & 874.84 & 0.27 & 0.50 \\\hline
\multicolumn{6}{c}{Expected Blocksize $\hat{b}=1$ years} \\ \hline
constant proportion($p=.5$) & 674 & 273 & 623.99 & 0.50 & 0.81 \\
NN adaptive & 942 & 459 & 878.60 & 0.27 & 0.50 \\\hline
\multicolumn{6}{c}{Expected Blocksize $\hat{b}=2$ years} \\ \hline
constant proportion($p=.5$) & 676 & 263 & 631.06 & 0.50 & 0.81 \\
NN adaptive & 945 & 438 & 882.74 & 0.26 & 0.50 \\\hline
\multicolumn{6}{c}{Expected Blocksize $\hat{b}=5$ years} \\ \hline
constant proportion($p=.5$) & 669 & 244 & 626.11 & 0.50 & 0.83 \\
NN adaptive & 940 & 404 & 881.87 & 0.23 & 0.50 \\\hline
\multicolumn{6}{c}{Expected Blocksize $\hat{b}=8$ years} \\ \hline
constant proportion($p=.5$) & 669 & 233 & 632.24 & 0.50 & 0.84 \\
NN adaptive & 945 & 388 & 892.84 & 0.22 & 0.50 \\\hline
\multicolumn{6}{c}{Expected Blocksize $\hat{b}=10$ years} \\ \hline
constant proportion($p=.5$) & 667 & 223 & 635.29 & 0.50 & 0.85 \\
NN adaptive & 942 & 373 & 895.88 & 0.22 & 0.50 \\\hline
\end{tabular}
\end{center}
\caption{Trained on bootstrap resampled data with $\hat{b}=1$ years} \label{tb_h:trn1robust}
}
\end{table}

\begin{table}[H]
{\scriptsize
\begin{center}
\begin{tabular}{lccccc} \hline
\multicolumn{6}{c}{ Test Results: Market Cap Weighted} \\ \hline
Strategy & $E(W_T)$ & $std(W_T)$& $median(W_T)$ & $Pr(W_T<median(W_T^{CP}))$ & $Pr(W_T<median(W_T^{NN}))$ \\ \hline
\multicolumn{6}{c}{Expected Blocksize $\hat{b}=0.5$ years} \\ \hline
constant proportion($p=.5$) & 678 & 286 & 623.07 & 0.50 & 0.83 \\
NN adaptive & 962 & 491 & 903.07 & 0.27 & 0.50 \\\hline
\multicolumn{6}{c}{Expected Blocksize $\hat{b}=1$ years} \\ \hline
constant proportion($p=.5$) & 674 & 273 & 623.99 & 0.50 & 0.83 \\
NN adaptive & 954 & 470 & 905.02 & 0.27 & 0.50 \\\hline
\multicolumn{6}{c}{Expected Blocksize $\hat{b}=2$ years} \\ \hline
constant proportion($p=.5$) & 676 & 263 & 631.06 & 0.50 & 0.84 \\
NN adaptive & 958 & 446 & 912.31 & 0.26 & 0.50 \\\hline
\multicolumn{6}{c}{Expected Blocksize $\hat{b}=5$ years} \\ \hline
constant proportion($p=.5$) & 669 & 244 & 626.11 & 0.50 & 0.85 \\
NN adaptive & 954 & 409 & 914.34 & 0.23 & 0.50 \\\hline
\multicolumn{6}{c}{Expected Blocksize $\hat{b}=8$ years} \\ \hline
constant proportion($p=.5$) & 669 & 233 & 632.24 & 0.50 & 0.87 \\
NN adaptive & 961 & 392 & 928.89 & 0.22 & 0.50 \\\hline
\multicolumn{6}{c}{Expected Blocksize $\hat{b}=10$ years} \\ \hline
constant proportion($p=.5$) & 667 & 223 & 635.29 & 0.50 & 0.88 \\
NN adaptive & 961 & 380 & 930.15 & 0.21 & 0.50 \\\hline
\end{tabular}
\end{center}
\caption{Trained on bootstrap resampled data with $\hat{b}=2$ years} \label{tb_h:trn2robust}
}
\end{table}

\begin{table}[H]
{\scriptsize
\begin{center}
\begin{tabular}{lccccc} \hline
\multicolumn{6}{c}{ Test Results: Market Cap Weighted} \\ \hline
Strategy & $E(W_T)$ & $std(W_T)$& $median(W_T)$ & $Pr(W_T<median(W_T^{CP}))$ & $Pr(W_T<median(W_T^{NN}))$ \\ \hline
\multicolumn{6}{c}{Expected Blocksize $\hat{b}=0.5$ years} \\ \hline
constant proportion($p=.5$) & 678 & 286 & 623.07 & 0.50 & 0.86 \\
NN adaptive & 995 & 495 & 963.03 & 0.26 & 0.50 \\\hline
\multicolumn{6}{c}{Expected Blocksize $\hat{b}=1$ years} \\ \hline
constant proportion($p=.5$) & 674 & 273 & 623.99 & 0.50 & 0.87 \\
NN adaptive & 988 & 478 & 963.28 & 0.25 & 0.50 \\\hline
\multicolumn{6}{c}{Expected Blocksize $\hat{b}=2$ years} \\ \hline
constant proportion($p=.5$) & 676 & 263 & 631.06 & 0.50 & 0.88 \\
NN adaptive & 994 & 458 & 973.65 & 0.25 & 0.50 \\\hline
\multicolumn{6}{c}{Expected Blocksize $\hat{b}=5$ years} \\ \hline
constant proportion($p=.5$) & 669 & 244 & 626.11 & 0.50 & 0.89 \\
NN adaptive & 997 & 427 & 976.51 & 0.22 & 0.50 \\\hline
\multicolumn{6}{c}{Expected Blocksize $\hat{b}=8$ years} \\ \hline
constant proportion($p=.5$) & 669 & 233 & 632.24 & 0.50 & 0.90 \\
NN adaptive & 1011 & 415 & 993.88 & 0.21 & 0.50 \\\hline
\multicolumn{6}{c}{Expected Blocksize $\hat{b}=10$ years} \\ \hline
constant proportion($p=.5$) & 667 & 223 & 635.29 & 0.50 & 0.92 \\
NN adaptive & 1015 & 409 & 996.57 & 0.20 & 0.50 \\\hline
\end{tabular}
\end{center}
\caption{Trained on bootstrap resampled data with $\hat{b}=5$ years} \label{tb_h:trn5robust}
}
\end{table}

\begin{table}[H]
{\scriptsize
\begin{center}
\begin{tabular}{lccccc} \hline
\multicolumn{6}{c}{ Test Results: Market Cap Weighted} \\ \hline
Strategy & $E(W_T)$ & $std(W_T)$& $median(W_T)$ & $Pr(W_T<median(W_T^{CP}))$ & $Pr(W_T<median(W_T^{NN}))$ \\ \hline
\multicolumn{6}{c}{Expected Blocksize $\hat{b}=0.5$ years} \\ \hline
constant proportion($p=.5$) & 678 & 286 & 623.07 & 0.50 & 0.86 \\
NN adaptive & 980 & 480 & 945.12 & 0.25 & 0.50 \\\hline
\multicolumn{6}{c}{Expected Blocksize $\hat{b}=1$ years} \\ \hline
constant proportion($p=.5$) & 674 & 273 & 623.99 & 0.50 & 0.86 \\
NN adaptive & 973 & 464 & 947.99 & 0.25 & 0.50 \\\hline
\multicolumn{6}{c}{Expected Blocksize $\hat{b}=2$ years} \\ \hline
constant proportion($p=.5$) & 676 & 263 & 631.06 & 0.50 & 0.87 \\
NN adaptive & 979 & 443 & 957.32 & 0.25 & 0.50 \\\hline
\multicolumn{6}{c}{Expected Blocksize $\hat{b}=5$ years} \\ \hline
constant proportion($p=.5$) & 669 & 244 & 626.11 & 0.50 & 0.88 \\
NN adaptive & 981 & 412 & 959.86 & 0.21 & 0.50 \\\hline
\multicolumn{6}{c}{Expected Blocksize $\hat{b}=8$ years} \\ \hline
constant proportion($p=.5$) & 669 & 233 & 632.24 & 0.50 & 0.90 \\
NN adaptive & 994 & 399 & 976.44 & 0.21 & 0.50 \\\hline
\multicolumn{6}{c}{Expected Blocksize $\hat{b}=10$ years} \\ \hline
constant proportion($p=.5$) & 667 & 223 & 635.29 & 0.50 & 0.91 \\
NN adaptive & 996 & 390 & 980.07 & 0.20 & 0.50 \\\hline
\end{tabular}
\end{center}
\caption{Trained on bootstrap resampled data with $\hat{b}=8$ years} \label{tb_h:trn8robust}
}
\end{table}

\begin{table}[H]
{\scriptsize
\begin{center}
\begin{tabular}{lccccc} \hline
\multicolumn{6}{c}{ Test Results: Market Cap Weighted} \\ \hline
Strategy & $E(W_T)$ & $std(W_T)$& $median(W_T)$ & $Pr(W_T<median(W_T^{CP}))$ & $Pr(W_T<median(W_T^{NN}))$ \\ \hline
\multicolumn{6}{c}{Expected Blocksize $\hat{b}=0.5$ years} \\ \hline
constant proportion($p=.5$) & 678 & 286 & 623.07 & 0.50 & 0.84 \\
NN adaptive & 963 & 468 & 920.86 & 0.25 & 0.50 \\\hline
\multicolumn{6}{c}{Expected Blocksize $\hat{b}=1$ years} \\ \hline
constant proportion($p=.5$) & 674 & 273 & 623.99 & 0.50 & 0.84 \\
NN adaptive & 957 & 451 & 923.63 & 0.25 & 0.50 \\\hline
\multicolumn{6}{c}{Expected Blocksize $\hat{b}=2$ years} \\ \hline
constant proportion($p=.5$) & 676 & 263 & 631.06 & 0.50 & 0.85 \\
NN adaptive & 962 & 431 & 932.13 & 0.25 & 0.50 \\\hline
\multicolumn{6}{c}{Expected Blocksize $\hat{b}=5$ years} \\ \hline
constant proportion($p=.5$) & 669 & 244 & 626.11 & 0.50 & 0.87 \\
NN adaptive & 962 & 399 & 937.08 & 0.22 & 0.50 \\\hline
\multicolumn{6}{c}{Expected Blocksize $\hat{b}=8$ years} \\ \hline
constant proportion($p=.5$) & 669 & 233 & 632.24 & 0.50 & 0.88 \\
NN adaptive & 973 & 384 & 951.40 & 0.21 & 0.50 \\\hline
\multicolumn{6}{c}{Expected Blocksize $\hat{b}=10$ years} \\ \hline
constant proportion($p=.5$) & 667 & 223 & 635.29 & 0.50 & 0.90 \\
NN adaptive & 973 & 373 & 954.63 & 0.20 & 0.50 \\\hline
\end{tabular}
\end{center}
\caption{Trained on bootstrap resampled data with $\hat{b}=10$ years} \label{tb_h:trn10robust}
}
\end{table}

\subsection{Robustness: Distribution Comparison Based on Test Results From the Synthetic Model} \label{append:distr}
We observe from Figure \ref{fig:syn_hist} that the terminal wealth distributions of the adaptive strategy are consistently right-skewed and have similar shapes in training and testing, which indicates that the NN strategy similarly outperforms the constant proportion in both training and testing.

\begin{figure}[htp]
\centering
\begin{subfigure}[b]{0.45\textwidth}
\centering
\includegraphics[width=\textwidth]{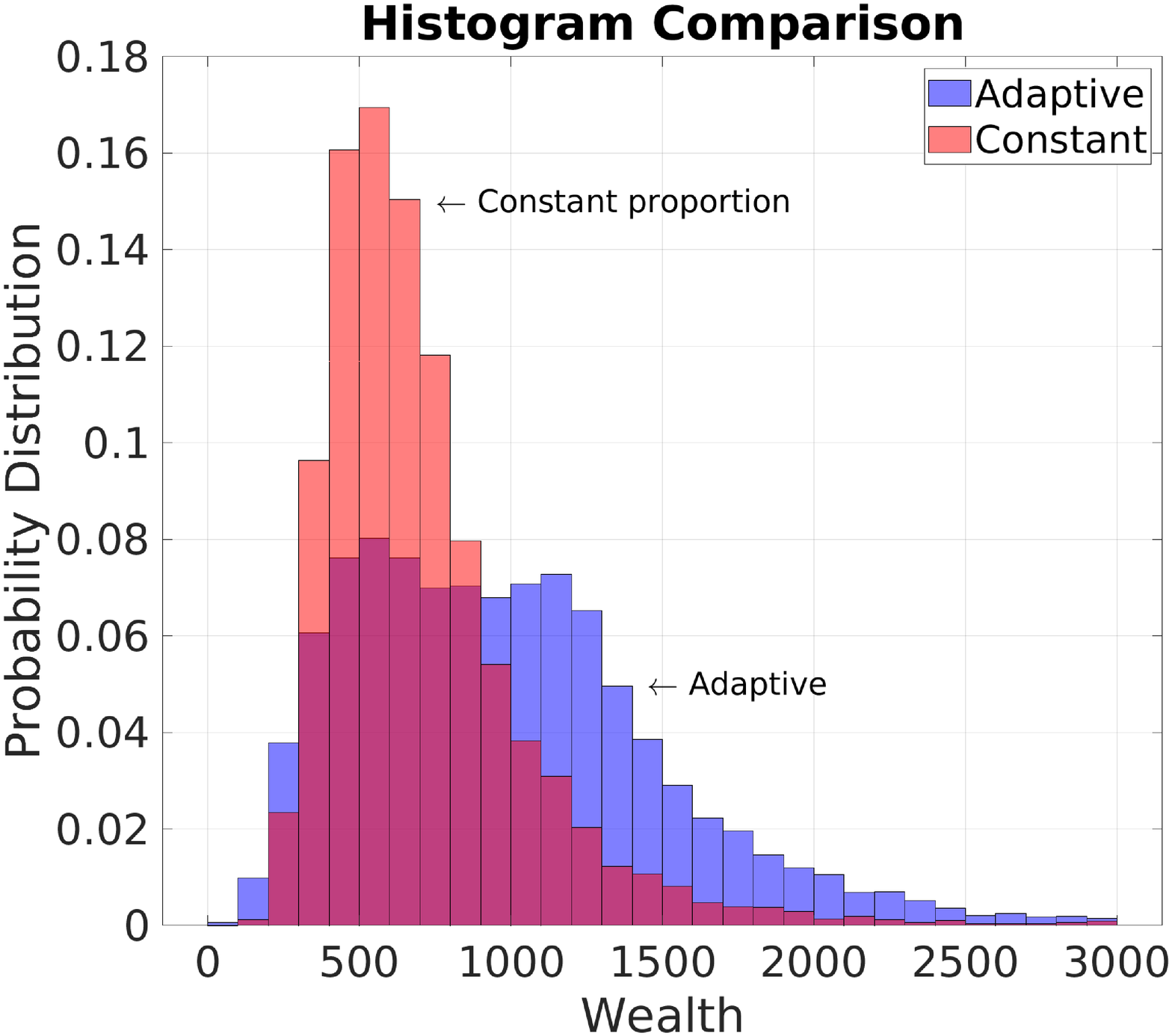}
\caption{Training on synthetics data}
\label{fig:syn_train_hist}
\end{subfigure}
\hfill
\begin{subfigure}[b]{0.45\textwidth}
\centering
\includegraphics[width=\textwidth]{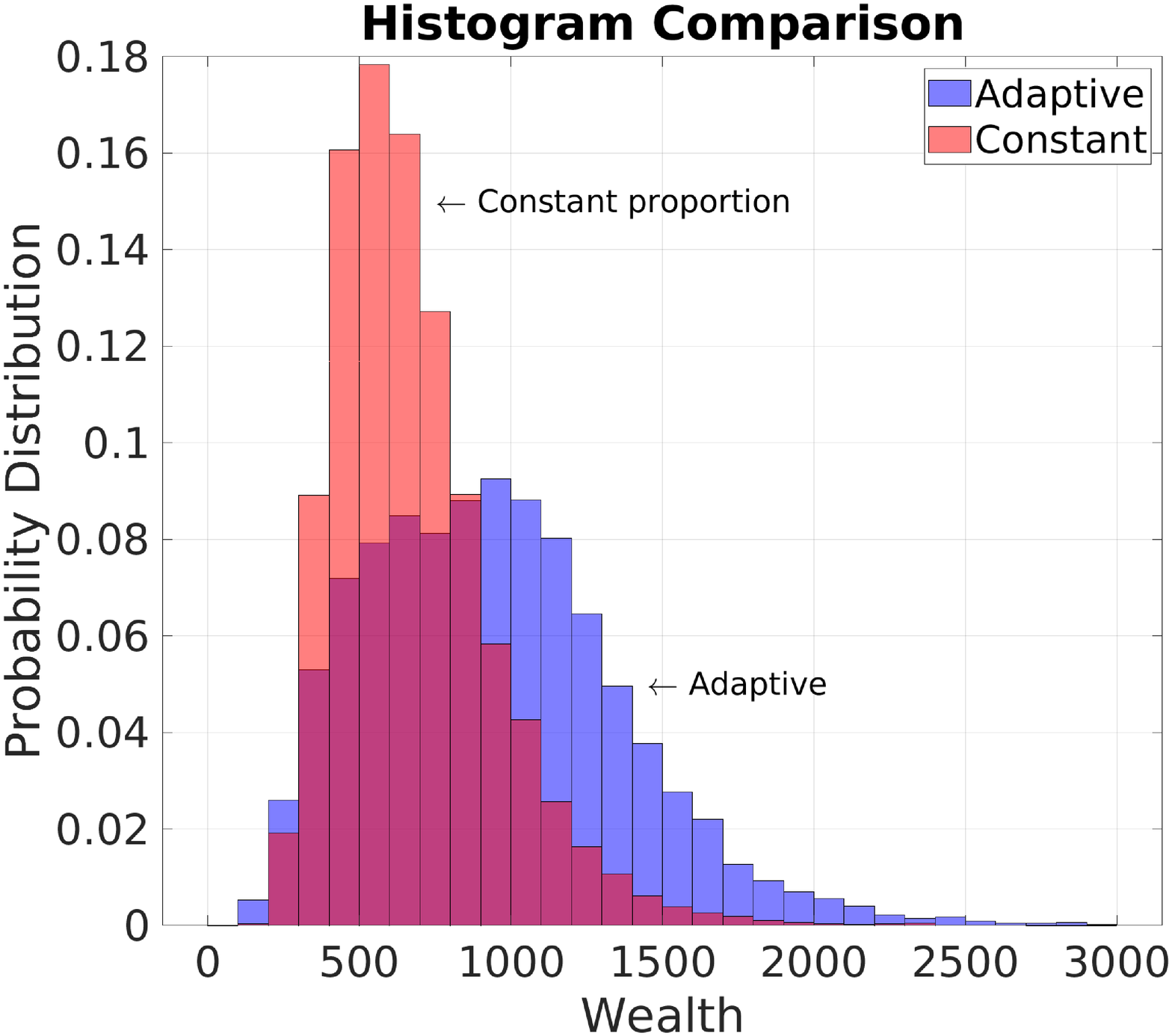}
\caption{Testing on bootstrap data with $\hat{b}$=0.5 years}
\label{fig:syn_test_hist}
\end{subfigure}
\caption{Histogram of terminal wealth. Model trained on synthetic data and tested on bootstrap resampled data with expected blocksize of $2$ years}
\label{fig:syn_hist}
\end{figure}

We also show the plot of the CDF of the wealth difference $W(T) - W_{50/50}(T)$ to give a more direct comparison between the adaptive strategy and constant proportion strategy on the same paths.

\begin{figure}[htp]
\centering
\begin{subfigure}[b]{0.45\textwidth}
\centering
\includegraphics[width=\textwidth]{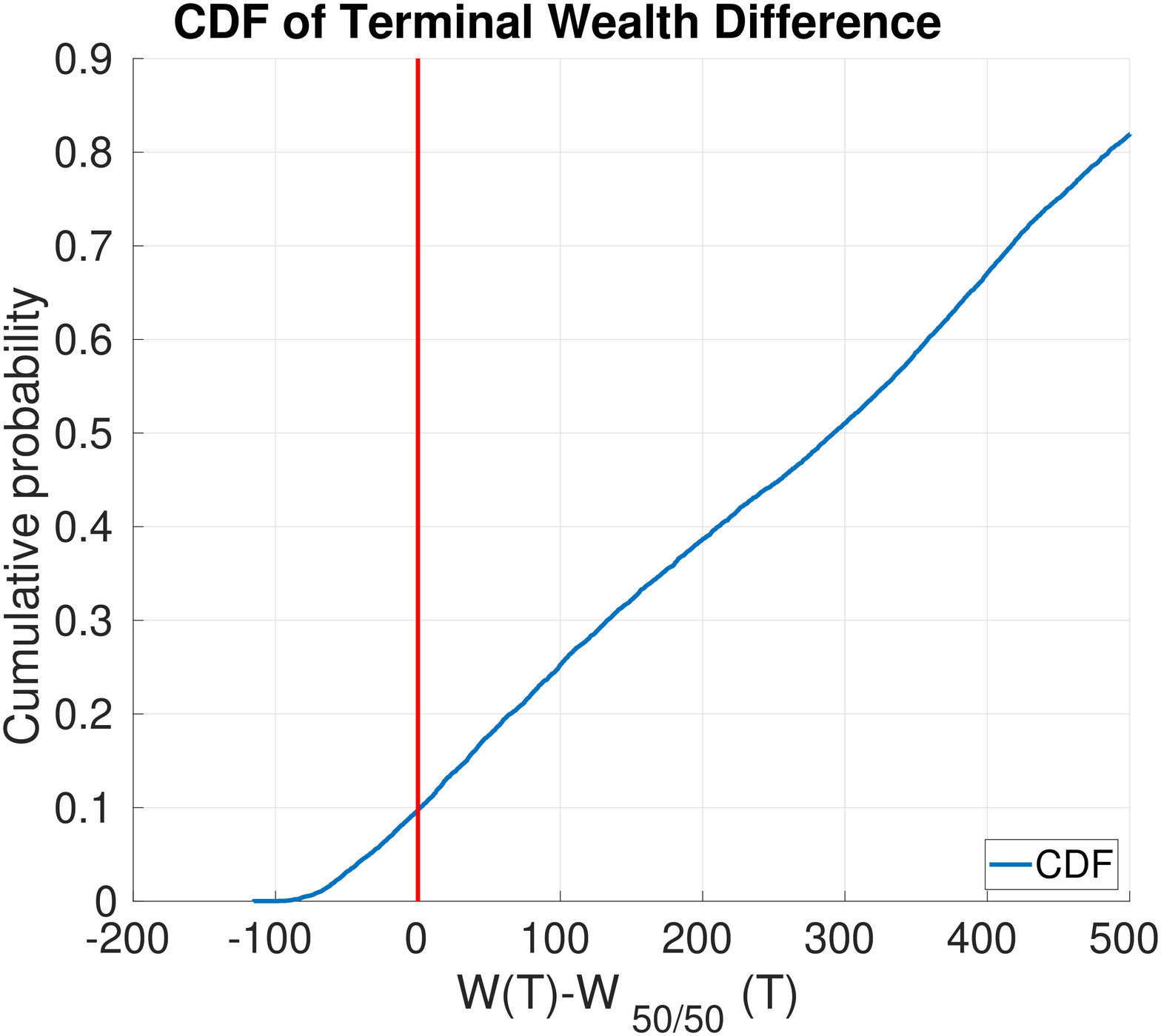}
\caption{Training on synthetics data}
\label{fig:syn_train_wdiff_cdf}
\end{subfigure}
\hfill
\begin{subfigure}[b]{0.45\textwidth}
\centering
\includegraphics[width=\textwidth]{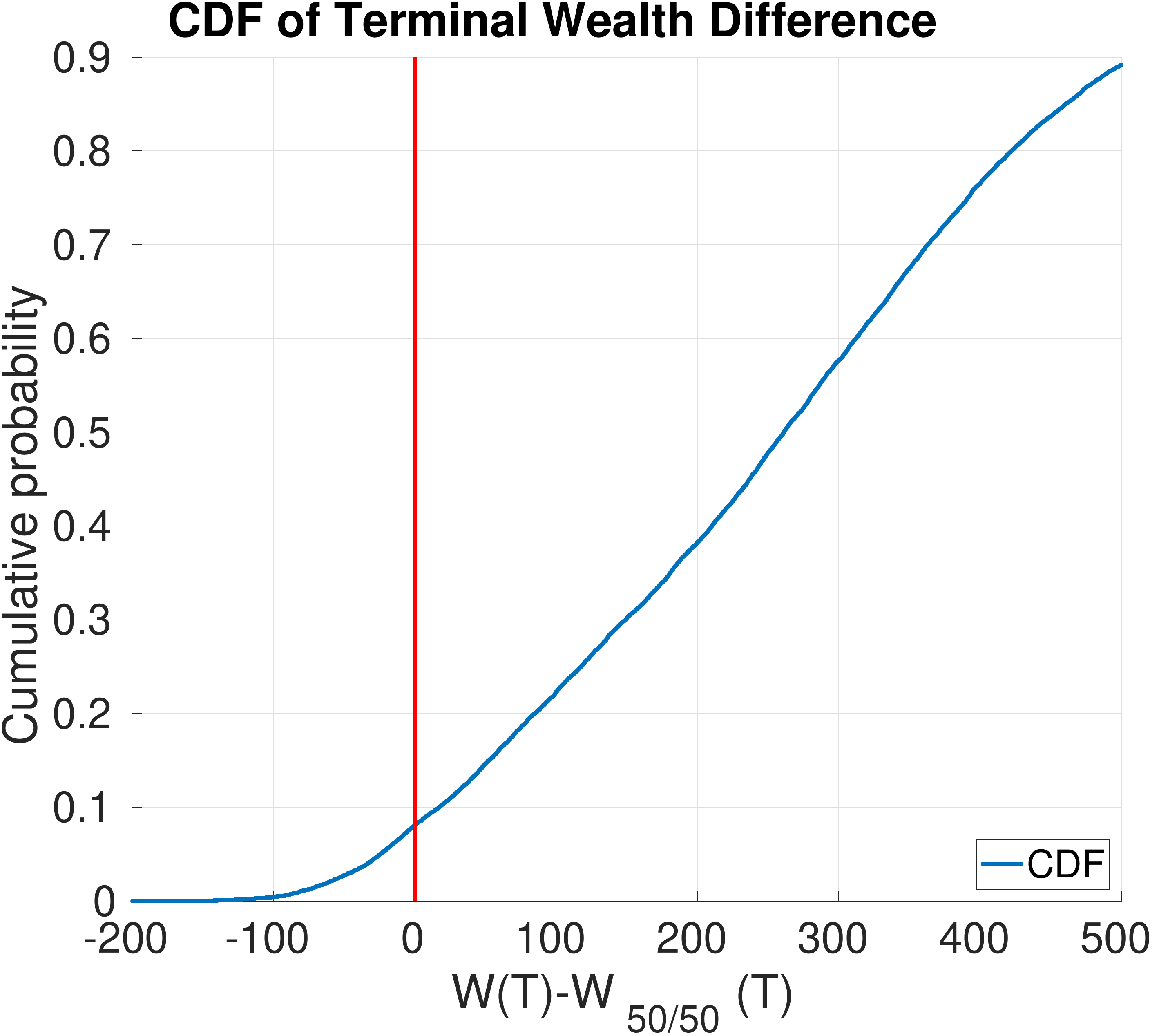}
\caption{Testing on bootstrap data with $\hat{b}$=2 years}
\label{fig:syn_test_wdiff_cdf}
\end{subfigure}
\caption{CDF of terminal wealth difference $W(T)-W_{50/50}(T)$}
\label{fig:syn_cdf}
\end{figure}

From Figure \ref{fig:syn_cdf} we can see that the probability of the adaptive strategy underperforming the constant proportion strategy is less than 10\% for both training and testing. When underperformance occurs, the scale of underperformance is small compared to the scale of potential outperformance. Therefore, we conclude that the adaptive strategy controls tail risks consistently in both training and testing, despite the fact that the training dataset is synthetically generated and the testing dataset is bootstrap resampled data.


\end{document}